%% file: main.tex
\documentclass[%
reprint,
superscriptaddress,
 amsmath,amssymb,
 aps
]{revtex4-2}

\usepackage{amsmath}
\usepackage{amsthm}
\newtheorem{theorem}{Theorem}
\newtheorem{lemma}{Lemma}

\usepackage{tcolorbox}
\usepackage{esint}
\usepackage{bbold}
\usepackage{amssymb}
\usepackage{graphicx}
\usepackage{textcomp}
\usepackage[margin=0.8in]{geometry}
\usepackage{relsize}
\usepackage{url}
\usepackage[colorlinks=true, allcolors=camblue]{hyperref}           
\usepackage{cancel}
\usepackage{tikz}
\usepackage{natbib}
\usepackage[utf8]{inputenc}
\usepackage[T1]{fontenc}
\usepackage[caption=false]{subfig}
\usetikzlibrary{quantikz2}

\usepackage[vlined]{algorithm2e}
\SetStartEndCondition{ }{}{}%
\SetKwProg{Fn}{def}{\string:}{}
\SetKwFunction{Range}{range}
\SetKw{KwTo}{in}\SetKwFor{For}{for}{\string:}{}%
\SetKwIF{If}{ElseIf}{Else}{if}{:}{elif}{else:}{}%
\SetKwFor{While}{while}{:}{fintq}%
\AlgoDontDisplayBlockMarkers
\SetAlgoNoEnd
\SetAlgoNoLine%
\SetKwFunction{FRecurs}{FnRecursive}%

\definecolor{camblue}{cmyk}{0.2443, 0.0000, 0.1250, 0.3098}

\makeatletter
\newcommand\niton{\mathrel{\m@th\mathpalette\canc@l\owns}}
\newcommand\canc@l[2]{{\ooalign{$\hfil#1/\mkern1mu\hfil$\crcr$#1#2$}}}
\makeatother

\usepackage{blkarray}
\usepackage{geometry} 

\usepackage{xcolor}
\usepackage{lipsum}

\begin{document}

\title{A Lightweight Protocol for Matchgate Fidelity Estimation}
\author{Jędrzej Burkat}
\email{jbb55@cam.ac.uk}
\affiliation{Cavendish Laboratory, Department of Physics, University of Cambridge, CB3 0HE, UK}
\author{Sergii Strelchuk}
\affiliation{Department of Computer Science, University of Oxford, OX1 3QD, UK}

\begin{abstract}
   
   
   

We present a low-depth randomised protocol for estimating the entanglement fidelity between an $n$-qubit matchgate circuit $\mathcal{U}$ and its noisy implementation $\mathcal{E}$. Our method uses a Pauli--Liouville representation adapted to Clifford algebra elements, in which matchgate superoperators acquire a block-diagonal form. This structure enables efficient direct fidelity estimation using only Pauli state preparation and measurement, while avoiding the exponentially costly sampling step required for generic unitary channels. Compared with the protocol of Flammia and Liu [PRL 106, 230501], our algorithm gives an exponential improvement in classical sampling complexity and a multiplicative $1/\sqrt{n}$ reduction in expected shot count for matchgate circuits. The protocol also extends, without asymptotic overhead, to matchgate circuits sandwiched between Clifford circuits. For nearest-neighbour $XY(\theta)$ gates and Givens rotations we demonstrate an increase in superoperator sparsity, giving an additional $1/\sqrt{n}$ reduction in expected shot count and, to our knowledge, the first scalable fidelity estimation protocol for these important matchgate subgroups.

\end{abstract}

\maketitle
\section{Introduction}
Implementing quantum computation is a complex and challenging process due to the ubiquitous presence of noise and decoherence. Accurately identifying the magnitude of these effects leads to the development of more efficient error mitigation strategies thus leading to better utilisation of existing quantum hardware.

Matchgates~\citep{Jozsa_2008, doi:10.1137/S0097539700377025, knill2001fermionic, Terhal_2002, DiVincenzo_2005, bravyi2008contraction} (MGs) represent an important family of two-qubit gates for quantum computation, forming an intriguing class of computation analogous to the dynamics of (non-interacting) free fermions, which are classically efficiently simulable. Similarly to the well-studied case of Clifford computation, matchgate circuits are classically simulable in polynomial time; however, with the addition of a $\mathrm{SWAP}$ gate their computational power is lifted to universal quantum computation. Unlike Clifford gates, they form an infinite, continuous (Lie) group of operators that can be implemented natively on a range of architectures via the M\o lmer-S\o rensen procedure~\citep{S_rensen_2000} and form key primitives in superconducting quantum computing \cite{Foxen_2020, sung2021realization, abrams2020implementation, Kjaergaard_2020}. They have also appeared in other architectures~\cite{PhysRevApplied.18.064082, grzesiak2020efficient}.

The continuous nature of matchgate circuits makes the study of their behaviour under realistic noise and decoherence conditions extraordinarily challenging -- and in many cases outright impossible. So far, two methods~\citep{Helsen_2022, PRXQuantum.2.010351} have been devised for experimentally estimating their fidelities, both relying on the exponential fitting of data from large-depth circuits to characterise their average noise behaviour. Notably, these approaches struggle with circuits composed of exclusively $XY(\theta)$ gates -- also known as anisotropic Heisenberg spin coupling \cite{kempe2001} or the $XY$ interaction. Such gates fall into the category of matchgates, forming a subgroup of the matchgate group.

We present an efficient protocol for fidelity estimation, similar in spirit to a widely used direct fidelity estimation (DFE) algorithm by S. Flammia and Y. Liu~\citep{Flammia_2011}. Compared to DFE on arbitrary unitary operations, our algorithm offers an exponential speedup in the classical sampling complexity (enabled by utilising an efficient determinantal point process sampling algorithm) and a multiplicative $1 / \sqrt{n}$ speedup in expected shot count for matchgate circuits. We first establish key mathematical results, working in a Pauli--Liouville (PL) representation of quantum channel adapted to matchgates. In this picture, it is experimentally straightforward to estimate the superoperator matrix overlap of two channels using only Pauli state preparation and measurement. Our procedure exhibits an additional $1/\sqrt{n}$ speedup for circuits composed of $XY(\theta)$ gates or Givens rotations, and may be extended to allow the benchmarking of matchgate circuits sandwiched by arbitrary, distinct Clifford unitaries. The algorithm is scalable, easy to implement in practice and the work therein has the potential to greatly simplify the process of calibrating matchgate circuits for use in quantum computation.

\section{Preliminaries}
\textit{Channel Fidelity.} Algorithms for estimating channel fidelities represent a key component of a variety of quantum information processing protocols. The channel fidelity describes the average overlap between pure states $| \psi \rangle \langle \psi |$ acted on by an idealised unitary channel $\mathcal{U}$, and the same state acted on by a (possibly noisy) quantum channel $\mathcal{E}$. Mathematically, it is defined as~\citep{Nielsen_2002}:

\begin{align}
    F(\mathcal{E}, \mathcal{U}) &= \int d \psi \langle \psi | U^{\dagger} \mathcal{E}(|\psi \rangle \langle \psi |) U | \psi \rangle \label{eq:fidelity_def}  \\
    & = \frac{d F_e(\mathcal{E}, \mathcal{U}) + 1}{d+1}. \nonumber
\end{align}
\noindent Where $d \psi$ is the uniform Haar measure over $\mathcal{H}$, normalised so that $\int d \psi = 1$, and $d$ is the Hilbert space dimension. $F_e(\mathcal{E}, \mathcal{U})$ is the entanglement fidelity -- a closely related and experimentally accessible quantity. Estimating it is the modus operandi of many randomised benchmarking protocols, as it allows for the direct calculation of $F(\mathcal{E}, \mathcal{U})$ by Equation \eqref{eq:fidelity_def}.

\textit{Pauli--Liouville (PL) Representation.} In the PL picture of quantum channels, we consider the vector space $\mathbb{C}^{2^{n} \times 2^n}$ spanned by the complex matrices $W_i \cong |i \rangle \rangle $, with the Hilbert-Schmidt inner product $\langle \langle i | j \rangle \rangle = 2^{-n} \text{Tr}(W_i^\dagger W_j)$. Under this representation, both pure and mixed quantum states may be written as $\rho = \sum_i \rho_i W_i$, or (using double ket notation) as $| \rho \rangle \rangle = \sum_i \rho_i | i \rangle \rangle$. The action of unitary operators, as well as CPTP maps, can then be represented through multiplication by a $2^{2n} \times 2^{2n}$ superoperator. Multiplication of vectors by superoperators $\hat{\Lambda}$ corresponds to the application of a quantum channel, whereas multiplying superoperator matrices gives a composition of quantum channels:
\begin{equation}
     \widehat{\Lambda \circ \Omega} | \rho \rangle \rangle = | \Lambda (\Omega (\rho)) \ \rangle \rangle = \hat{\Lambda} \hat{\Omega  } | \rho \rangle \rangle.
\end{equation}
In a given basis, the superoperators can be written as:
\begin{equation}
    \hat{\Lambda} = \sum_{i,j=1}^{2^{2n}} \chi_\Lambda(i,j) | i \rangle \rangle \langle \langle j |,
\end{equation}
with the matrix elements given by:
\begin{equation}
    \chi_\Lambda(i,j) = \frac{1}{2^n}\text{Tr}(W_i^\dagger \Lambda(W_j)).
\end{equation}

Our protocol will make use of the Clifford basis, a set of $2^{2n}$ orthonormal monomials of Clifford algebra generators $c_i$, the latter forming a set of $2n$ Hermitian matrices satisfying the anti-commutation relations:
\begin{equation}
    \{c_i, c_j \} = 2 \delta_{ij} \mathbb{1}, \ \ \ \ \ i, j \in [2n].
\end{equation} Up to unitary equivalence, the generators can be chosen to be the Jordan--Wigner operators:

\begin{align}
    c_{2k-1} &= Z \otimes \cdots \otimes Z \otimes X \otimes \mathbb{1} \otimes \cdots \otimes \mathbb{1}  \\
    c_{2k} &= \underbrace{Z \otimes \cdots \otimes Z}_{k-1} \otimes \ Y \otimes \underbrace{\mathbb{1} \otimes \cdots \otimes \mathbb{1}}_{n-k}. \nonumber
\end{align}

\noindent We will denote elements of this basis as $c_I = c_{i_1} \hdots c_{i_{|I|}}$, where $I \subseteq \{ 1, \hdots, 2n \}$ is a set of lexicographically ordered indices taken from $[2n]$, with $c_{\emptyset} = \mathbb{1}$ (an example of our ordering can be seen in Figure \ref{fig:fsim_diagram}). In the Jordan--Wigner representation every basis element is proportional to a Pauli string:
\begin{equation}
    c_I = \phi_I \mathbb{P}_I, \label{clifford_and_pauli}
\end{equation}
and for any subset $I \subseteq [2n]$ with cardinality $k$ the Hermitian conjugate is given by $c_I^\dagger = (-1)^{k(k-1)/2} c_I$. As a result, the phase $\phi_I$ is either $\pm 1$ or $\pm i$, and remains real or pure imaginary across all monomials indexed by subsets with constant degree $k$ (see Appendix \ref{appendixa}). The difference between Clifford and Pauli bases then amounts to a re-ordering of the basis elements and change of phase -- a subtle alteration which considerably simplifies the protocol.

\textit{Matchgate PL Superoperators.} Matchgates (also known as fermionic linear optics gates) are two-qubit operators -- denoted $G(A, B)$ -- acting as $A \in U(2)$ on the even-parity $\{ |00 \rangle, |11 \rangle \}$, and $B \in U(2)$ on the odd-parity $\{ |01 \rangle, |10 \rangle \}$ subspaces of $\mathbb{C}^2 \otimes \mathbb{C}^2$: 

\begin{equation}
    G(A,B) = \begin{pmatrix}
        a_{11} & 0 & 0 & a_{12} \\
        0 & b_{11} & b_{12} & 0 \\
        0 & b_{21} & b_{22} & 0 \\
        a_{21} & 0 & 0 & a_{22}
      \end{pmatrix}, \label{eq:matchgate}
\end{equation}

\noindent such that $\det A = \det B$. Operators of this form for which $\det A \neq \det B$ (for example, a $\mathrm{SWAP}$ or $\mathrm{CZ}$ gate) are denoted $\tilde{G}(A, B)$. 

Any matchgate circuit (MGC) may be written as a Gaussian operation $U = e^{-iH}$, where:
\begin{equation}
    H = i \sum_{i \neq j} h_{i j} c_i c_j \label{eq:hamiltonian}
\end{equation}
is a quadratic Hamiltonian, with the coefficients $h_{i j}$ forming a real, antisymmetric $2n \times 2n$ matrix. Previous work~\citep{knill2001fermionic,Terhal_2002,DiVincenzo_2005,Jozsa_2008} has shown that the orthogonal $SO(2n)$ matrix $R = e^{4h}$ generated from $h$ provides full information about $U$. Furthermore, conjugating a generator $c_i$ by a matchgate unitary $U$ is a rotation by $R$ in the space of Clifford algebra generators:
\begin{equation}
    U c_i U^\dagger = \sum_{j = 1}^{2n} R_{j i} c_j.
\end{equation}    
This relation suggests that the Clifford basis is a natural choice for the Pauli--Liouville representation of matchgate unitaries. It turns out that elements of the associated superoperator matrix $\hat{\mathcal{U}}$ in this basis can be efficiently calculated, only requiring access to $R$.

\begin{theorem} \label{thm:matchgatesupop}
    For any matchgate circuit $U = e^{-iH}$ with an associated $SO(2n)$ matrix $R$, in the Clifford algebra generator basis the superoperator matrix elements $\chi_\mathcal{U}(I,J)$ are given by:
    \begin{equation}
        \chi_\mathcal{U}(I, J) = \delta_{|I|, |J|} \det| R_{I,J} |,
    \end{equation}
    where $R_{I, J}$ is a submatrix of $R$ for which each $i$\textsuperscript{th} row $i \in I$ is preserved, and each $j$\textsuperscript{th} column $j \in J$ is preserved. Consequently, $\hat{\mathcal{U}}$ is a block diagonal matrix with $2n+1$ blocks, each of which is a ${2n \choose k} \times {2n \choose k}$ compound matrix $C_k(R)$, where $k = |I|$ is the degree of the monomial $c_I = c_{i_1} \hdots c_{i_k}$.
\end{theorem}
\begin{proof}
    This is a known and straightforward result (see \cite{Chapman_2018}, for example). Starting with the definition of $\chi_\mathcal{U}(I,J)$, we can prove it by inserting $\mathbb{1} = U^\dagger U$ between each $c_j$ operator:
\begin{widetext}
    \begin{align}
        \chi_\mathcal{U} (I, J) &= \frac{1}{2^n} \text{Tr} \left( (c_{i_1} \hdots c_{i_m})^\dagger U (c_{j_1} \hdots c_{j_k}) U^\dagger  \right) \label{eq:matchgate_superop} \\
        &= \frac{1}{2^n} \text{Tr} \left( (c_{i_1} \hdots c_{i_m})^\dagger U c_{j_1} U^\dagger \hdots U c_{j_k} U^\dagger  \right) \nonumber \\
        &= \frac{1}{2^n} \text{Tr} \left( (c_{i_1} \hdots c_{i_m})^\dagger \sum_{\substack{\nu_1, \hdots, \nu_k \\ \nu_1 \neq \hdots \neq \nu_k}} R_{\nu_1 j_1} \hdots R_{\nu_k j_k} c_{\nu_1} \hdots c_{\nu_k}\right) \nonumber \\
        &= \frac{1}{2^n} \text{Tr} \left( (c_{i_1} \hdots c_{i_m})^\dagger \mathop{\sum \ \sum}_{\substack{1 \leq \nu_1 < \hdots < \nu_k \leq 2n, \\ \sigma \in S_k}} \text{sgn}(\sigma) R_{\nu_{\sigma(1)} j_1} \hdots R_{\nu_{\sigma(k)} j_k} c_{\nu_1} \hdots c_{\nu_k} \right) \nonumber \\
        &= \delta_{mk} \sum_{\sigma \in S_k} \text{sgn}(\sigma) R_{i_{\sigma(1)} j_1} R_{ i_{\sigma(2)} j_2} \hdots R_{i_{\sigma(k)} j_k} \nonumber \\
        &= \delta_{mk} \det |R_{(i_1 \hdots i_k, j_1 \hdots j_k)}|. \nonumber \qedhere
    \end{align}
\end{widetext}
\end{proof}

\begin{tcolorbox}[title = {Matchgate Fidelity Estimation: Requirements and Runtime}, colback=white, colframe=camblue, fonttitle=\bfseries, title filled=false]
    \textit{Inputs:} 
    \begin{itemize}
        \item Quantum circuit $\mathcal{E}$ for an idealised matchgate unitary $\mathcal{U}$, with the ability to prepare and measure Pauli eigenstates. 
        \item $SO(2n)$ matrix $R$ associated with $\mathcal{U}$, from which superoperator matrix elements $\chi_\mathcal{U}(I,J)$ can be efficiently computed.
        \item \textit{(Optional)} Well-conditioning parameter $\alpha$, such that for all $(I, J)$, either $|\chi_\mathcal{U}(I,J)| \geq \alpha$ or $|\chi_\mathcal{U}(I,J)| = 0$.
    \end{itemize}
    \textit{Output:} 
    \begin{itemize}
        \item Entanglement fidelity estimator $\tilde{Y}$, for which $F_e(\mathcal{E}, \mathcal{U})$ lies in the range $[\tilde{Y} - 2\epsilon, \tilde{Y} + 2\epsilon]$ with probability at least $1 - 2\delta$. 
    \end{itemize}

    \textit{Runtime Parameters (without \& with well-conditioning):} 
    \begin{itemize}
        \item Sample numbers: 
        
        $l = \left\lceil \frac{1}{\epsilon^2 \delta} \right\rceil, \ \left \lceil \frac{2 \ln(2/\delta)}{\alpha^2 \epsilon^2} \right \rceil $
        \item Iteration number: 
        
        $m_\mu = \left\lceil \frac{2 \ln(2/\delta)}{\chi_\mathcal{U}(I_\mu, J_\mu)^2 l \epsilon^2 } \right\rceil$

        \item Total shot bounds: 
        
        $\mathbb{E}(m) \leq \mathcal{O}(\frac{1}{\epsilon^2 \delta} + \frac{2^{2n} \ln(1/\delta)}{\sqrt{n} \epsilon^2}), \\ m \leq \mathcal{O}(\frac{\ln(1/\delta)}{\alpha^2 \epsilon^2}) $ 

        \item Classical processing time: $\mathcal{O}(n^3)$ 
    \end{itemize}

    \textit{For any $\hat{\mathcal{U}}$ with PL matrix elements $\chi_\mathcal{U}(I,J)$:}
    \begin{center}
    $\mathbb{E}(m) \leq \mathcal{O}(\frac{1}{\epsilon^2 \delta} + \frac{\#[\chi_\mathcal{U}(I,J)  \neq 0]}{2^{2n}}\frac{\ln(1/\delta)}{\epsilon^2})$
    \end{center}
\end{tcolorbox}

Theorem \ref{thm:matchgatesupop} tells us that in the PL representation a matchgate superoperator can have at most $\sum_{k=0}^{2n}{2n \choose k}^2 = {4n \choose 2n}$ elements. Asymptotically, this upper bounds the number of non-zero elements as:
\begin{equation}
    \#[\chi_\mathcal{U}(I,J) \neq 0 ] \leq 2^{4n} / \sqrt{n}. \label{eq:mg_nonzero_bound}
\end{equation}
\noindent In Section \ref{sec:xygroup} we show that for circuits composed of exclusively $XY(\theta)$ gates or Givens rotations, this is further reduced to $\#[\chi_\mathcal{U}(I,J) \neq 0 ] \leq 2^{4n} / n$. For comparison, Clifford gates form monomial matrices with entries equal to $\pm 1$ or $\pm i$. For Clifford-matchgate circuits the superoperator matrix becomes block-diagonal again, with each block forming an orthogonal monomial matrix (i.e. a signed permutation in $B(m) \subset SO(m)$). 

\section{Matchgate Fidelity Estimation}
Let $\mathcal{U}$ be a unitary channel (with an associated superoperator $\hat{\mathcal{U}}$) describing the action of an $n$-qubit matchgate circuit $U$, and $\mathcal{E}$ the mathematical description of its actual, noisy implementation. We assume the ability to either prepare and measure states in the Pauli eigenbasis or apply single-qubit unitaries out of $\{X, H, S \}$ to the $| 0 \rangle$ state with negligible error (which itself may be characterised using other methods first), and wish to quantify the distance between $\mathcal{U}$ and $\mathcal{E}$. A natural measure for this is the entanglement fidelity, which can be calculated from the superoperator  overlap:

\begin{align}
    F_e(\mathcal{E}, \mathcal{U}) &= \frac{1}{2^{2n}} \text{Tr}(\hat{\mathcal{U}}^\dagger \hat{\mathcal{E}}) \label{eq:entanglement_fidelity} \\
   &= \frac{1}{2^{2n}} \sum_{i,j,k,l} \chi^*_\mathcal{U}(i,j) \chi_\mathcal{E}(k,l) \langle \langle i | k \rangle \rangle \langle \langle l| j \rangle \rangle \nonumber \\
   &= \frac{1}{2^{2n}} \sum_{i, j} \chi^*_\mathcal{U}(i, j) \chi_\mathcal{E}(i, j). \nonumber
\end{align}

To estimate the entanglement fidelity for matchgate circuits, we modify the DFE algorithm \citep{Flammia_2011} to use the Clifford basis: first, pick pairs of indices $(I, J)$ at random with probabilities $\text{Pr}(I, J) = 2^{-2n} |\chi_\mathcal{U}(I, J)|^2$ and then determine matrix elements $\chi_\mathcal{E}(I, J)$, thus constructing the random variable: 

\begin{equation}
    X = \frac{\chi_\mathcal{E}(I, J)}{\chi_\mathcal{U}(I, J)}, \label{eq:X_variable}
\end{equation}

\noindent which satisfies $\mathbb{E}[X] = F_e(\mathcal{E}, \mathcal{U})$.  For the case of matchgates, by Equation \eqref{eq:matchgate_superop} all superoperator matrix elements are real, and vanishing whenever $|I| \neq |J|$. Thus, Equation \eqref{eq:entanglement_fidelity} may be rewritten as:

\begin{equation}
    F_e(\mathcal{E}, \mathcal{U}) = \frac{1}{2^{2n}}\sum_{k=0}^{2n} \ \sum_{|I| = |J| = k} \chi_\mathcal{E}(I, J) \chi_\mathcal{U}(I, J). \label{eq:mg_fidelity_sum}
\end{equation}

Before applying our algorithm, one must isolate the matrix $R$ associated with the matchgate circuit (c.f. Equation \eqref{eq:hamiltonian}). In contrast to DFE on general unitaries, our approach does not require pre-computing the full superoperator matrix, and thus all classical steps (including sampling $(I, J)$, see Section \ref{sec:sup_sampling}) can be performed efficiently. Key information concerning our protocol is provided on the previous page and Algorithms \ref{box:algorithm}, \ref{box:algorithm2}. Proofs are relegated to Appendix \ref{appendixb} and \ref{app:dpp}. 

\textit{Expectation Values.} To illustrate why our protocol works, we follow similar reasoning to \cite{Flammia_2011}.  We consider the `perfect' random variable $X$, as in \eqref{eq:X_variable}, which may be averaged to reveal the entanglement fidelity. We define an estimator $Y = \frac{1}{l} \sum_\mu X_\mu$, where $l$ is the `sample number' and $X_\mu$ is the $\mu^\text{th}$ sample of $X$, obtained by randomly choosing a pair of indices $(I_\mu, J_\mu)$ (c.f. step 1). In expectation, $\mathbb{E}(Y) = \mathbb{E}(X_\mu)=F_e(\mathcal{E}, \mathcal{U})$, as the expectation value of $X_\mu$ over the random choices of $(I_\mu, J_\mu)$ is given by:

\begin{align}
    \mathbb{E}(X_\mu) &= \sum_\mu \frac{\chi_\mathcal{E}(I_\mu, J_\mu)}{\chi_\mathcal{U}(I_\mu, J_\mu)} \frac{\chi_\mathcal{U}(I_\mu, J_\mu)^2}{2^{2n}} \\
    &= F_e(\mathcal{E}, \mathcal{U}). \nonumber
\end{align}

Because direct measurement of superoperator elements $\chi_\mathcal{E}(I_\mu, J_\mu)$ is not possible, each $X_\mu$ must be approximated with an estimator $\tilde{X}_\mu$, obtained from performing $m_\mu$ `iterations' of a random circuit (steps 3--5) and averaging the collected data (step 6). Equation \eqref{eq:bmunu} shows that over measurement outcomes, the quantity $B_{\mu \nu}$ (constructed in step 5) averages to $\chi_\mathcal{E}(I_\mu, J_\mu)$. Therefore, the expectation value of each $\tilde{X}_\mu$ is:
\begin{equation}\mathbb{E}( \tilde{X}_\mu ) = \frac{\chi_\mathcal{E}(I_\mu, J_\mu)}{\chi_\mathcal{U}(I_\mu, J_\mu)} = X_\mu. \end{equation}

Finally, we define a third estimator $\tilde{Y} = \frac{1}{l} \sum_\mu \tilde{X}_\mu$, which satisfies $\mathbb{E}(\tilde{Y}) = Y$ and is experimentally accessible through steps 1-7. Algorithm \ref{box:algorithm} may thus be seen as the composition of two successive estimations, returning an estimate $\tilde{Y}$ of $Y$. In turn, $Y$ is an estimate of the entanglement fidelity. In expectation we have:
\begin{equation}
    \mathbb{E}_{\mu} [ \mathbb{E}_\nu (\tilde{Y}) ] = \mathbb{E}_\mu (Y) = F_e(\mathcal{E}, \mathcal{U}).
\end{equation}

\

\begin{tcolorbox}[title = Algorithm 1: Matchgate Fidelity Estimation, colback=white, colframe=camblue, fonttitle=\bfseries, title filled=false]
    \begin{algorithm}[H]
        \DontPrintSemicolon
        \For{$\mu$ \KwTo\Range{$l$}} {
            \ \\
            1. Sample a pair of index subsets 
            $(i_1 \hdots i_k, j_1 \hdots j_k)_\mu = (I_\mu, J_\mu)$
             via Algorithm \ref{box:algorithm2}, with probability:
            \begin{equation}
            \text{Pr}(I, J) = 2^{-2n} \left( \chi_\mathcal{U}(I, J) \right)^2. \nonumber \end{equation}
            2. Convert the corresponding Clifford monomials into Pauli operators as:
            \begin{align}
                (c_{i_1} \hdots c_{i_k})_\mu ^\dagger &= \phi^*_{I_\mu} \mathbb{P}_{I_\mu}, \nonumber \\ (c_{j_1} \hdots c_{j_k})_\mu &= \phi_{J_\mu} \mathbb{P}_{J_\mu}, \nonumber
            \end{align}
            and store the product of phases as:
            \begin{equation} \phi_{\mu} = \phi^*_{I_\mu} \times \phi_{J_\mu} = \pm 1. \nonumber \end{equation}
            \For{$\nu$ \KwTo\Range{$m_\mu$}} {
            3. Prepare a random eigenstate $| \lambda_{\mu \nu} \rangle$ of $\mathbb{P}_{J_\mu}$ with probability $2^{-n}$, storing its eigenvalue $\lambda_{\mu \nu} = \pm 1$. \;
            4. Apply the matchgate circuit $\mathcal{E}$. \;
            5. Make single-qubit Pauli ($\mathbb{P}_{I_\mu}$) measurements on $q \in \mathrm{supp}(\mathbb{P}_{I_\mu})$ to obtain outcomes $c_q = \pm 1$. Compute $A_{\mu \nu}$, where:
            \begin{equation}
            A_{\mu \nu} = \Pi_{q \in \mathrm{supp}(\mathbb{P}_{I_\mu})} c_q, \nonumber \end{equation}
            and compute the product $B_{\mu \nu}$:
            \begin{equation}
            B_{\mu \nu} = A_{\mu \nu} \lambda_{\mu \nu} \phi_{\mu} = \pm 1. \nonumber \end{equation}
            }
            6. Compute the estimator $\tilde{X}_\mu$:
            \begin{equation}
                \tilde{X}_\mu = \frac{1}{\chi_\mathcal{U}(I_\mu, J_\mu) m_\mu} \sum_{\nu = 1}^{m_\mu} B_{\mu \nu}. \nonumber \end{equation}
        }
        7. Compute the estimator $\tilde{Y}$: 
        \begin{equation*}
        \tilde{Y} = \frac{1}{l}\sum_{\mu = 1}^l \tilde{X}_\mu = \sum_{\mu = 1}^l \sum_{\nu = 1}^{m_\mu} \frac{A_{\mu \nu} \lambda_{\mu \nu} \phi_{\mu}}{\chi_\mathcal{U}(I_\mu, J_\mu) m_\mu l}. \end{equation*}
    \KwRet $\tilde{Y}$
    \label{box:algorithm}{}
    \end{algorithm}
\end{tcolorbox}

\begin{widetext}
    \begin{align}
        \mathbb{E}(B_{\mu \nu}) &= \sum_{k=1}^{2^n} \frac{1}{2^n} \lambda^{(k)}_{\mu \nu} \phi_\mu \left( \text{Pr}(A^{(k)}_{\mu \nu} = +1) - \text{Pr}(A^{(k)}_{\mu \nu} = -1) \right) \label{eq:bmunu} \\
        &= \frac{\phi_\mu}{2^n} \sum_{k=1}^{2^n} \lambda^{(k)}_{\mu \nu} \  \text{Tr} \left( \mathbb{P}^\dagger_{I_\mu} \mathcal{E}(| \lambda^{(k)}_{\mu \nu} \rangle \langle \lambda^{(k)}_{\mu \nu} |) \right) \nonumber \\
        &= \frac{\phi_\mu}{2^n} \text{Tr} \left( \mathbb{P}^\dagger_{I_\mu} \mathcal{E}(\mathbb{P}_{J_\mu}) \right) \nonumber \\
        & = \chi_\mathcal{E}(I_\mu, J_\mu). \nonumber
    \end{align}
\end{widetext}

Steps 3--5 (repetitions over $\nu$) bring $\tilde{Y}$ to within $\pm \epsilon$ of $Y$ (with failure probability $\delta$), whereas taking samples of $\tilde{X}_\mu$ (repetitions over $\mu$) brings $Y$ $\epsilon$-close to $F_e(\mathcal{E}, \mathcal{U})$. Overall, we obtain an estimate of the entanglement fidelity with a probability $1 - 2\delta$ of lying within $\pm 2\epsilon$ of its true value.

\textit{Well-Conditioning.} We give two bounds for $l$: one for the general case, where no assumption is made about the structure of $\hat{\mathcal{U}}$ or the size of its elements, and another for the case where $|\chi_\mathcal{U}(I, J)| \geq \alpha$ for all non-zero elements. The latter is the `well-conditioning' property from \cite{Flammia_2011}. Whilst it is true that all Clifford gates are well-conditioned (with $\alpha = 1$), a matchgate circuit can correspond to an arbitrarily small rotation with a small value of $\alpha$. Determining this value for large systems is an intractable problem, as it requires sorting exponentially many non-zero elements of $\hat{\mathcal{U}}$. For smaller systems, we provide guidance on how to isolate the well-conditioning parameter in Appendix \ref{appendixc}. 

\section{Efficient Superoperator Index Sampling} \label{sec:sup_sampling}

A major bottleneck in applying direct fidelity estimation to arbitrary unitaries is the classical sampling \cite{Barber_2025} of indices $(I, J)$, i.e. step 1 of Algorithm \ref{box:algorithm}. In general, this requires (i) classically pre-computing $2^{4n}$ matrix elements $\chi_\mathcal{U}(I, J)$, and (ii) sampling from the resulting (exponentially large) probability distribution $\text{Pr}(I, J) = 2^{-2n} |\chi_\mathcal{U}(I, J)|^2$. Furthermore, as the quantity $ 2^{-2n} |\chi_\mathcal{U}(I, J)|^2$ is the mod-square of a normalised trace (the estimation of which is \textsf{DQC1}-complete \cite{Knill_1998, shepherd2006}) performing the sampling step for generic unitaries is likely intractable. Here, we show that when $\mathcal{U}$ is a matchgate unitary, the classical sampling can be performed efficiently. Thus, the classical component of Algorithm \ref{box:algorithm} is efficient, achieving the claimed exponential speedup over the original DFE protocol.

\begin{lemma}
For any matchgate circuit $\mathcal{U}$, given $R \in SO(2n)$, each superoperator element $\chi_\mathcal{U}(I, J)$ can be computed exactly in $\mathcal{O}(n^3)$ time.  
\end{lemma}
\begin{proof}
From Theorem \ref{thm:matchgatesupop}, we have that $\chi_\mathcal{U}(I, J) = \delta_{|I|, |J|}\det|R_{I, J}|$, where $R_{I, J}$ is the submatrix of $R$ formed by taking rows in $I$ and columns in $J$. Therefore, each $R_{I, J}$ is a $k \times k$ matrix, where $k = |I| = |J| \leq 2n$. Thus, the determinant can be computed in $\mathcal{O}(k^3) \leq \mathcal{O}(n^3)$ time using standard methods. 
\end{proof}

\begin{theorem} \label{thm:matchgate_sampling}
For any matchgate circuit $\mathcal{U}$, classically sampling from the probability distribution $\text{Pr}(I, J) = 2^{-2n} |\chi_\mathcal{U}(I, J)|^2$ can be achieved in $\mathcal{O}(n^3)$ time.
\end{theorem}
\begin{proof}
First, we note that the superoperator matrix $\hat{\mathcal{U}}$ is block diagonal, with blocks corresponding to monomial degrees $k = |I| = |J|$. We begin by calculating the squared-element sum $S_k$ for a degree-$k$ block:
\begin{equation}
    S_k = \sum_{|I|=|J|=k} |\chi_\mathcal{U}(I, J)|^2 = \sum_{|I|=|J|=k} |\det(R_{I, J})|^2.
\end{equation}
By the Cauchy-Binet formula (Theorem \ref{thm:compound_matrices}), the elements $\det(R_{I, J}) = [C_k(R)]_{I, J} \in SO(k)$ form an orthogonal compound matrix. Therefore, we may write:
\begin{align}
S_k &=  \sum_{I, J} [C_k(R)]^2_{I, J} \\
&= \text{Tr}(C_k(R)^\dagger C_k(R)) \nonumber\\
&= \text{Tr}(C_k(R^\dagger R)) = \binom{2n}{k}. \nonumber
\end{align}

\noindent We thus begin by drawing a degree $k$ with probability $\text{Pr}(k) = S_k / 2^{2n}$, i.e. $k \sim \mathrm{Binomial}(2n, 1/2)$. Next, for a sampled degree $k$, we draw the row index $I$ with uniform probability among the $\binom{2n}{k}$ subsets of $[2n]$ of size $k$. To see that this is a uniform distribution, note that:
\begin{align}
\text{Pr}(I | k) &= \frac{1}{S_k} \sum_{|J|=k} |\chi_\mathcal{U}(I, J)|^2 \\
& = \frac{1}{S_k} \sum_{|J|=k} [C_k(R)]^2_{I, J} \nonumber\\
& = \frac{1}{S_k} \sum_{|J|=k} [C_k(R)^\dagger C_k(R)]_{I, I} \nonumber\\
& = \frac{1}{S_k} [C_k(\mathbb{1}_{2n})]_{I, I} = \frac{1}{S_k}.\nonumber
\end{align}

Finally, for a sampled index $I$ of degree $k$, we draw the column index $J$ with probability $\text{Pr}(J | I, k) = |\chi_\mathcal{U}(I, J)|^2$, i.e. sample $J$ from the $I^\text{th}$ row of the degree-$k$ block of $\hat{\mathcal{U}}$. This distribution is still factorially large; we now show that it can be sampled from efficiently using a determinantal point process (DPP) \cite{Kulesza_2012, Hough_2006, derezinski2019, derezinski2020}.

We begin by constructing a $k \times 2n$ matrix $R_{I, [2n]}$, composed of rows in $R$ contained in $I$. Next, construct a $2n \times 2n$ matrix $K^{(I)}= R_{I, [2n]}^T R_{I, [2n]}$ which is real, symmetric, and idempotent: 
\begin{align}
(K^{(I)})^2 &= R_{I, [2n]}^T \left( R_{I, [2n]} R_{I, [2n]}^T \right) R_{I, [2n]} \\
&= R_{I, [2n]}^T \mathbb{1}_k R_{I, [2n]} = K^{(I)}, \nonumber
\end{align}

\noindent where the second line follows from the fact that rows of $R_{I, [2n]}$ are orthonormal, so $ R_{I, [2n]} R_{I, [2n]}^T = \mathbb{1}_k$. In fact, $K^{(I)}$ is a projection matrix onto the $k$-dimensional row space of $R_{I, [2n]}$, guaranteeing it to be positive semidefinite. For a $k$-sized subset $J \subseteq [2n]$ one may construct the $k \times k$ submatrix $K_{J, J}^{(I)}$ by preserving the rows and columns in $J$. Observe that:

\begin{align}
[K^{(I)}]_{j_1, j_2} &= [R_{I, [2n]}^T R_{I, [2n]}]_{j_1, j_2} \\
& = \sum_{i \in I} [R_{I, [2n]}^T]_{j_1, i} [R_{I, [2n]}]_{i, j_2} \nonumber \\
& = \sum_{i \in I} [R_{I, [2n]}]_{i, j_1} [R_{I, [2n]}]_{i, j_2} \nonumber,
\end{align}

\noindent from which it follows that $K_{J, J}^{(I)} = R_{I, J}^T R_{I, J}$, where $R_{I, J}$ is the $k \times k$ submatrix of $R$ as defined in Theorem \ref{thm:matchgatesupop}. Therefore, we have:

\begin{align}
    \det(K_{J, J}^{(I)}) &= \det(R_{I, J}^T R_{I, J}) \\
    &= \det(R_{I, J}^T) \det(R_{I, J}) \nonumber \\
    &= |\det(R_{I, J})|^2 = |\chi_\mathcal{U}(I, J)|^2. \nonumber
\end{align}

\noindent Because $K^{(I)}$ is symmetric, real and positive semidefinite, we may sample $J$ from the distribution $\text{Pr}(J | I, k) = \det(K_{J, J}^{(I)})$ using a DPP sampling algorithm \cite{Hough_2006, derezinski2019} with the kernel $K^{(I)}$. Overall:
\begin{equation}
\text{Pr}(J | I, k) = \det(K_{J, J}^{(I)}) = |\chi_\mathcal{U}(I, J)|^2. \label{eq:dpp_determinant}
\end{equation} 

\noindent Notably, the DPP algorithm runs in $\mathcal{O}(n^3)$ time \cite{Hough_2006}, whereas sampling from the binomial distribution for $k$ and the uniform distribution for $I$ can be achieved in $\mathcal{O}(n)$ time. More information on DPPs and the sampling algorithm is provided in Appendix \ref{app:dpp}.
\end{proof}

Overall, to sample $(I,J)$ we may proceed as follows: first set $I = \emptyset$, and perform $2n$ fair coin flips, setting $I \leftarrow I \cup \{c\}$ if the $c^\text{th}$ coin flip is heads. This gives $\mathrm{Pr}(I, k) = \mathrm{Pr}(I | k) \mathrm{Pr}(k) = 2^{-2n}$. Then, construct the matrix $K^{(I)}$ for the sampled $I$, and use it as the kernel for DPP sampling (Algorithm \ref{box:algorithm3} in Appendix \ref{app:dpp}). Each step runs in at most $\mathcal{O}(n^3)$ time, so the overall sampling procedure is efficient. A summary of our procedure is presented in Algorithm \ref{box:algorithm2}.

\

\begin{tcolorbox}[title = Algorithm 2: Superoperator Index Sampling, colback=white, colframe=camblue, fonttitle=\bfseries, title filled=false]
    \begin{algorithm}[H]
        \DontPrintSemicolon
        1. Set $I = \emptyset$. \\
        \For{$c$ \KwTo\Range{$2n$}} {
            2. Draw $q \sim \text{Bernoulli}(1/2)$. \\
             \If{$q = 1$}{3. Set $I \leftarrow I \cup \{c\}$.} }
        4. Construct the DPP kernel $K^{(I)}$: 
        \begin{equation} K^{(I)} = R_{I, [2n]}^T R_{I, [2n]}. \nonumber \end{equation} \\
        5. Draw $J \sim \text{DPP}(K^{(I)})$ via Algorithm \ref{box:algorithm3}. \\
        \KwRet $(I, J)$ 
    \label{box:algorithm2}{}
    \end{algorithm}
\end{tcolorbox}

    \begin{center}
        \defcitealias{Helsen_2022}{HNRW'20}
        \defcitealias{PRXQuantum.2.010351}{CRW'20}
        \begin{figure*}[t]
        \begin{tabular}{| c || c | c | c |} 

         \hline
          &  \textbf{Present Work} & \textbf{\citetalias{Helsen_2022}} & \textbf{\citetalias{PRXQuantum.2.010351}} \\ [0.5ex] 
         \hline\hline

         \textbf{Outcome:} & \footnotesize $F(\mathcal{E,U})$: the fidelity between $\mathcal{U}$ & \multicolumn{2}{|c|}{ \footnotesize The fidelity $F(\mathcal{E})$ of the average error channel $\mathcal{E}$, acting as} \\
         & \footnotesize  and its noisy implementation $\mathcal{E}$. & \multicolumn{2}{|c|}{ \footnotesize $\mathcal{E} \circ \mathcal{U}$ on all MGC channels $\mathcal{U}$. Gives the bound $F(\mathcal{\mathcal{E} \circ \mathcal{U}}) \leq F(\mathcal{E})$.}  \\ [0.5ex]
         \hline

         \textbf{Assumptions:} & \footnotesize Noiseless Pauli basis SPAM & \multicolumn{2}{|c|}{\footnotesize Robust against SPAM noise.}  \\
         & \footnotesize (State Preparation and Measurement). & \multicolumn{2}{|c|}{\footnotesize Assumes a gate-independent noise channel $\mathcal{E}$.} \\ [0.5ex]
         \hline

         \textbf{Applied} & \footnotesize Matchgate circuit, & \multicolumn{2}{|c|}{ \footnotesize Matchgate circuits with Haar random} \\
         \textbf{Gates:} & \footnotesize Pauli SPAM. & \multicolumn{2}{|c|}{ \footnotesize $R \in SO(2n)$, Pauli SPAM.} \\ [0.5ex]
         \hline

         \textbf{Noise} & \footnotesize PL Superoperator elements $\chi_\mathcal{E}(I, J):$ & \multicolumn{2}{|c|}{\footnotesize Decay parameters $\lambda_k$, $k \in [2n]$:}  \\
         \textbf{Parameters:}& \footnotesize $\chi_\mathcal{E}(I, J) = \text{Tr}(c_I^\dagger \mathcal{E}(c_J)) / 2^n$. & \multicolumn{2}{|c|}{\footnotesize $\lambda_k = \text{Tr}(\hat{P}_k \hat{\mathcal{E}}) / {2n \choose k}$ for irrep projectors $\{\hat{P}_0, ..., \hat{P}_{2n} \}$.} \\
         &\footnotesize $F_e(\mathcal{E, U}) = 2^{-2n} \sum_{I, J} \chi^*_\mathcal{U}(I, J) \chi_\mathcal{E}(I, J).$ & \multicolumn{2}{|c|}{ \footnotesize $F_e(\mathcal{E}) = 2^{-2n}  \sum_{k=0}^{2n} {2n \choose k} \lambda_k$.} \\ [0.5ex]
         \hline

         \textbf{Noise} & \footnotesize Prepare random eigenstates of $c_J$,& \footnotesize $2n + 1$ polynomial & \footnotesize $n$ polynomial fittings\\
         \textbf{Parameter} & \footnotesize apply the channel $\mathcal{E}$ and measure & \footnotesize fittings of & \footnotesize of $f_k(N) = \sum_{l=1}^2 C_{kl} \lambda^N_{kl}$,\\
         \textbf{Estimation:} & \footnotesize in the $c^\dagger_I$ basis. Repeat many times. & \footnotesize $f_k(N) = C_k \lambda_k^N$.&\footnotesize one fitting of $f_0(N) = C_0 \lambda^N_0$.\\ [0.5ex]
         \hline

         \textbf{Procedure:} & \footnotesize Algorithm \ref{box:algorithm}. & \multicolumn{2}{|c|}{ \footnotesize (i) Pick $k \in [2n]$ (or $k \in [n]$). (ii) Prepare a Pauli eigenstate, } \\
         & & \multicolumn{2}{|c|}{ \footnotesize apply $\mathcal{O}(N)$ randomly sampled MGs, measure in Pauli basis.} \\
         & & \multicolumn{2}{|c|}{ \footnotesize (iii) Repeat (ii) $M$ times. (iv) Repeat (ii) - (iii) for different $N$. } \\
         & & \multicolumn{2}{|c|}{ \footnotesize (v) Repeat (i) - (iv) for different $k$. (vi) Compute functions $\hat{f}_k(N)$} \\ 
         & & \multicolumn{2}{|c|}{ \footnotesize from measurement data, and polynomially fit to find the decay} \\
         & & \multicolumn{2}{|c|}{ \footnotesize parameters $\lambda_k$. (vii) Estimate $F(\mathcal{E})$ using Equation \eqref{eqn:helsen_fidelity}. } \\ [0.5ex]
         \hline

         \textbf{Runtime} & \footnotesize $l$ samples (choices of SPAM basis), & \footnotesize \ \ \ \ $2n+1$ values of $k \in [2n],$ \ \ \ \ & \footnotesize $n + 1$ values of $k \in [n]$,  \\
         \textbf{Parameters:} &\footnotesize $m_\mu$ iterations (random circuit & \footnotesize repetitions $M$, & \footnotesize repetitions $M$, \\
         & \footnotesize repetitions). & \footnotesize sequence lengths $N$. & \footnotesize sequence lengths $N$. \\ [0.5ex]
         \hline 
        
         \textbf{Shot Scaling:} & \footnotesize $\mathbb{E}(m) \leq \mathcal{O}(1 / \epsilon^2 \delta + 2^{2n} \ln(1/\delta) / \sqrt{n} \epsilon^2)$ & \multicolumn{2}{|c|}{\footnotesize No analytic bounds provided.} \\ [0.5ex]
            \hline

        \ \textbf{Sampling Cost:} \ & \footnotesize $\mathcal{O}(n^3)$ & \multicolumn{2}{|c|}{\footnotesize $\mathcal{O}(\mathrm{poly}(n))$} \\ [0.5ex]
            \hline

         \textbf{Benchmarking} & \footnotesize Same procedure. Shot bound decreases: & \multicolumn{2}{|c|}{\footnotesize Benchmark $XX(\theta)$ gates via $XX(\theta) = XY(\theta / 2) X_1 XY(\theta / 2) X_1$,} \\
         \textbf{XY Gates:} & \footnotesize $\mathbb{E}(m) \leq \mathcal{O}(1 / \epsilon^2 \delta + 2^{2n} \ln(1/\delta) / n \epsilon^2)$ & \multicolumn{2}{|c|}{\footnotesize Assume $X_1$ is noiseless and $XY(\theta/2)$ noise is multiplicative.} \\[0.5ex]
         \hline

        \end{tabular}
        \caption{\textit{Comparison of Matchgate Benchmarking Protocols.} $F_e(\mathcal{E, U})$ and $F_e(\mathcal{E}) = F_e(\mathcal{E}, \mathbb{1} )$ are the entanglement fidelities, related to channel fidelity via Equation \eqref{eq:fidelity_def}.}
        \label{table:1}
        \end{figure*}
    \end{center}

\section{Comparison to existing methods}

Prior to our work, two methods for benchmarking matchgates have been published; one due to J. Claes et al.~\citep{PRXQuantum.2.010351}, and the other due to J. Helsen et al.~\citep{Helsen_2022}. The two methods are similar in their approach, making use of the group structure of nearest-neighbour $G(A,B)$ gates via randomised benchmarking (RB). This is achieved by repeated preparation of random circuits, with random elements of the matchgate group applied at variable circuit depths. After collecting measurements, data is polynomially fitted to extract error parameters $\lambda$, which are then used to estimate the fidelity of the gate set. These approaches assume that the gate noise being measured is independent of the particular gates $U \in G$ applied, i.e. constant across the matchgate group (when this assumption is relaxed, the deviations from the gate-independent noise model are suppressed, so the resulting noise information may still be seen as an average). Comparisons between the three methods are given in Figure \ref{table:1}. 

A recent addition to the literature is the work of A. Chapman and S. Flammia \cite{chapman2025}, which proposes an efficient method for learning parameters of group-averaged, gate-dependent noise channels for collections of matchgates. Similarly to RB-based methods this approach utilises matchgate twirling and consequently returns average channel parameters, although the noise is allowed to vary between different group elements.

In contrast, our procedure allows one to directly benchmark each instance of $U$ without averaging over $G$. It does not require polynomial fitting to estimate the relevant error parameters -- the entanglement fidelity estimate $F_e(\mathcal{E, U})$ is given directly from the constructed $\tilde{Y}$. By virtue of Theorem \ref{thm:matchgate_sampling} the sampling component is also simpler, never requiring drawing from the Haar measure. However, our approach has lower robustness to SPAM noise, though only Pauli state preparation and measurement are required. We believe our method is advantageous in its simplicity and speed, and will benefit practitioners who wish to benchmark gates at some particular parameter $\theta$, for example the $XX(\theta)$ or $XY(\theta)$ gates in the context of calibration. Our approach is also scalable, and if $\alpha$ is additionally identified then the shot number bound is no longer dependent on system size $n$, opening up the possibility of fast benchmarking of matchgate circuits on many qubits.

\subsection{Decay Parameters \texorpdfstring{$\lambda$}{lambda} and the Fidelity}

The guiding principles between RB-based techniques and our work are broadly similar. The goal of randomised benchmarking is to identify the average fidelity of a constant (or average) error channel $\mathcal{E}$:

\begin{equation}
F(\mathcal{E}) = F(\mathcal{E}, \mathbb{1}) = \frac{2^n F_e(\mathcal{E}, \mathbb{1}) + 1}{2^n + 1}.
\end{equation}

\noindent The entanglement fidelity $F_e(\mathcal{E}, \mathbb{1})$ is then the trace of the error channel:

\begin{equation}
    F_e(\mathcal{E}, \mathbb{1}) = \frac{1}{2^{2n}}\text{Tr}(\hat{\mathcal{E}}).
\end{equation}

\noindent The authors of \citep{Helsen_2022} make use of projection matrices $\{ \hat{P}_k \}$, which project onto the invariant subspaces $\{ \Gamma_k \}$ of irreducible representations of the matchgate group. Inserting $\sum_{k=0}^{2n} \hat{P}_k = \hat{\mathbb{1}}$ into the above, they obtain:

\begin{align}
    \frac{1}{2^{2n}}\text{Tr}(\hat{\mathcal{E}}) &= \frac{1}{2^{2n}} \sum_{k=0}^{2n} \text{Tr}(\hat{P}_k \hat{\mathcal{E}}) \label{eqn:helsen_fidelity}  \\
    &= \frac{1}{2^{2n}} \sum_{k=0}^{2n} {2n \choose k} \lambda_k, \nonumber 
\end{align}

\noindent where:

\begin{equation}
    \lambda_k = \frac{\text{Tr}(\hat{P}_k \hat{\mathcal{E}})}{\text{Tr}(\hat{P}_k)} = \frac{1}{{2n \choose k}} \text{Tr}(\hat{P}_k \hat{\mathcal{E}}).
\end{equation}

\noindent The subspaces projected onto by $\hat{P}_k$ are exactly the block-diagonals of $\hat{\mathcal{U}}$ as described in Theorem \ref{thm:matchgatesupop}. Comparing our expression to Equation \eqref{eq:mg_fidelity_sum}, we can relate the decay parameters $\lambda_k$ to the partial sums:

\begin{align}
    \lambda_k' &= \frac{1}{{2n \choose k}} \sum_{|I| = |J| = k} \chi_\mathcal{E}(I, J) \chi_\mathcal{U}(I, J)  \\
    &= \frac{1}{{2n \choose k}} \text{Tr}(\hat{P}_k \hat{\mathcal{U}}^\dagger \hat{\mathcal{E}}), \nonumber
\end{align}

\noindent which evaluate to $F_e(\mathcal{E, U})$ through Equation \eqref{eqn:helsen_fidelity}. Our procedure finds `gate-specific' decay parameters $\lambda_k'$, which encode information about the fidelity between a particularly chosen MGC and its noisy implementation. This is in contrast to the decay parameters of \citep{Helsen_2022, PRXQuantum.2.010351,chapman2025} which characterise the noise channel under a matchgate twirl.

\begin{figure*}[t!]
    \subfloat{ \includegraphics[width=.5\linewidth]{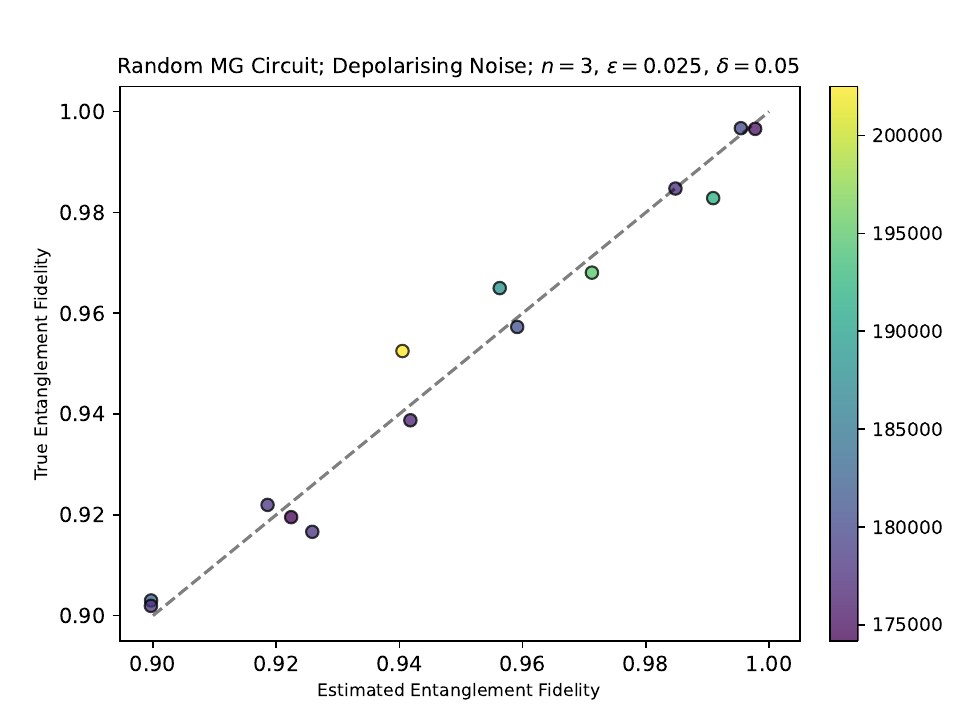} }
    \subfloat{ \includegraphics[width=.5\linewidth]{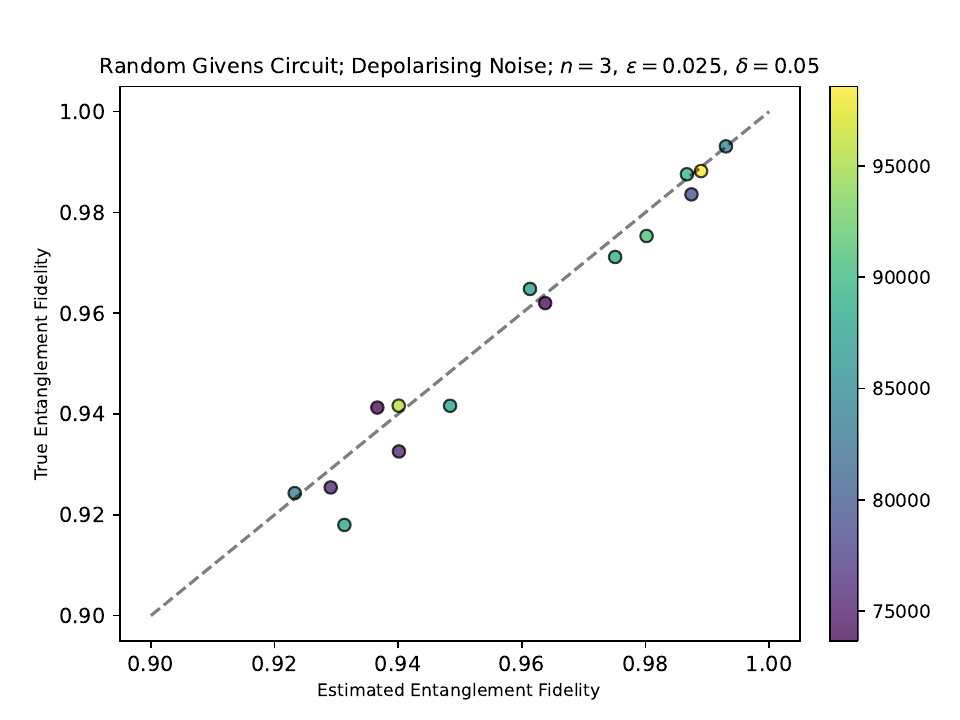} }
    \caption{\textit{Simulation of Algorithm \ref{box:algorithm} for $n=3$ Qubits.} Each point is the output of Algorithm \ref{box:algorithm} applied to a Haar-randomly sampled circuit followed by an $n$-qubit depolarising channel. Colour maps indicate shot numbers in each run. Restricting to nearest-neighbour Givens rotations is seen to decrease the shot number by about $1/\sqrt{3}$, in agreement with Section \ref{sec:xygroup}. Code used to simulate our algorithm is provided in \cite{mgrepo}.}
    \label{fig:simulations1}
\end{figure*}

\subsection{The XY Group}

\label{sec:xygroup}

Despite their technical sophistication, RB-based approaches tend to struggle with subgroups of the matchgate group (for example, $XY(\theta)$ gates or Givens rotations). In general, the group formed by matchgate circuits has irreps appearing with multiplicities, which complicate identifying the decay parameters $\lambda_k$. The authors of \citep{Helsen_2022} avoid this issue by expanding the group to include a bit-flip $X$ on the final qubit, obtaining a multiplicity-free decomposition. They then propose benchmarking $XY(\theta)$ gates indirectly, by preparing $XX$ gates through $XX(\theta) = XY(\theta / 2) (X \otimes \mathbb{1}) XY(\theta / 2) (X \otimes \mathbb{1})$, and assuming that (i) the single-qubit gates are noiseless and (ii) the noise from the $XY$ gates is multiplicative.
On the other hand, the method in \cite{PRXQuantum.2.010351} handles multiplicities directly by fitting functions of $a_k$ degree-$N$ polynomials $\sum_i^k C_{ij} \lambda_{ij}^N$ (where $a_k$ is the multiplicity of the $k^\text{th}$ irrep) to the data. For the matchgate group, these multiplicities are at most $2$, though for the $XY$ subgroup they grow with the qubit number $n$. The authors find that because of this scaling, their proposed method of multiplicity benchmarking fails for the $XY$ group -- so neither is feasible for directly estimating the fidelity of $XY(\theta)$ gates.

In contrast, our method is able to benchmark $XY(\theta)$ circuits without additional difficulty, with an asymptotic $1/\sqrt{n}$ speedup compared to general MGCs. This is because for such circuits, the rotation matrix $R$ takes the form (under reshuffling of rows and columns):

\begin{equation}
    R \cong \tilde{R} \oplus \tilde{R}', \quad \tilde{R} \in SO(n), 
\end{equation}

\noindent where $\tilde{R}'_{ij} = (-1)^{i+j} \tilde{R}_{ij}$. Similarly, for $U$ composed of nearest-neighbour Givens rotations exclusively:

\begin{equation}
R = \tilde{R} \otimes \mathbb{1}_2, \quad \tilde{R} \in SO(n).
\end{equation}

\noindent These observations follow from considering the quadratic Hamiltonians of nearest-neighbour $XY$ gates and Givens rotations, which generate dynamical Lie algebras that split into two copies of $\mathfrak{so}(n)$ (see \cite{K_kc__2022, Wiersema_2024}). In fact, for any matchgate subgroup for which $R \cong A \oplus B$ for $A, B \in SO(n)$ we can show that the non-zero element bound reduces by a factor of $1/\sqrt{n}$. First, write $I = I_A \cup I_B$ and $J = J_A \cup J_B$, where $I_A, J_A \subseteq \{1, \ldots, n\}$ and $I_B, J_B \subseteq \{n+1, \ldots, 2n\}$. Any degree-$k$ minor $\det(R_{I, J})$ can only be non-zero if $|I_A| = |J_A| = m$ and $|I_B| = |J_B| = k - m$, otherwise $R_{I, J}$ would contain linearly dependent rows or columns. Then, the number of non-zero elements in $\hat{\mathcal{U}}$ is upper bounded by:
\begin{align}
    \#[\chi_\mathcal{U}(I,J) \neq 0] & \leq \sum_{k=0}^{2n} \sum_{m = 0}^{k} {n \choose m}^2 {n \choose k - m}^2 \\
    &\leq \sum_{m = 0}^{n} \sum_{l = 0}^{n} {n \choose m}^2 {n \choose l}^2 \nonumber\\
    &\leq {2n \choose n}^2 \nonumber \\
    &\leq (2^{2n} / \sqrt{n})^2 \leq 2^{4n} / n, \nonumber
\end{align}

\noindent which is a multiplicative $1/\sqrt{n}$ reduction compared to the generic matchgate case (Equation \eqref{eq:mg_nonzero_bound}), leading to the same reduction in expected shot count (c.f. Equation \eqref{eqn:totexpshotbound}) The full upper bound is given in Figure \ref{table:1}.

\section{Beyond Matchgates: Benchmarking Clifford + MG + Clifford circuits}
\label{sec:cliffordmg}
Our protocol may also be used (without additional overhead) to benchmark circuits of the form $\mathcal{W} = \mathcal{V}_2 \circ \mathcal{U} \circ \mathcal{V}_1$, where $\mathcal{U}$ is a matchgate and $\mathcal{V}_1, \mathcal{V}_2$ are arbitrary Clifford circuits. Because Clifford unitaries form monomial superoperator matrices, one can calculate $\hat{\mathcal{W}}$ directly and proceed with the same instructions ($\hat{\mathcal{W}}$ will have entries from $\hat{\mathcal{U}}$ permuted and multiplied by a phase from $\{ \pm 1, \pm i\}$, preserving its sparsity). However, explicit calculation of $\hat{\mathcal{W}}$ is not necessary. For the new circuit, pairs of $c_{I'}, c_{J'}$ are sampled using the probability distribution:
\begin{equation}
    \text{Pr}(I', J') = 2^{-2n}(\text{Tr}(c_{I'}^\dagger V_2 U V_1 c_{J'} V_1^\dagger U^\dagger V_2^\dagger))^2,
\end{equation}

\noindent which is equal to the original distribution $\text{Pr}(I, J)$ for $c_I = V_2^\dagger c_{I'} V_2$ and $c_J = V_1 c_{J'} V_1^\dagger$. Hence, we may use the original probability distribution, and upon sampling $I, J$ calculate $c_{I'} = V_2 c_I V_2^\dagger$, $c_{J'} = V_1^\dagger c_J V_1$. In the $c_{I'}$ basis the superoperator elements remain invariant:
\begin{align}
    \chi_\mathcal{W}(I', J') &= \frac{1}{2^n} \text{Tr}(c_{I'}^\dagger V_2 U V_1 c_{J'} V_1^\dagger U^\dagger V_2^\dagger)  \\
    &= \frac{1}{2^n} \text{Tr}(c_I^\dagger U  c_J  U^\dagger)\nonumber \\
    &= \chi_\mathcal{U}(I, J). \nonumber
\end{align}

Because the transformed operators $c_{I'}$ and $c_{J'}$ are also Pauli strings, we may estimate $\chi_\mathcal{E}(I',J')$ by preparing eigenstates of $c_{J'}$, applying $\mathcal{W}$, and measuring in the $c_{I'}$ basis. This is effectively a passive transformation on the original algorithm, introducing no further computational cost other than identifying the new Pauli strings $c_{I'}$ and $c_{J'}$ (which may be done efficiently using stabilizer simulation methods). This is summarised by the following modification of Algorithm \ref{box:algorithm}:

\begin{itemize}
    \item[${2^*}.$] \textit{If the MGC is preceded by a Clifford circuit $V_1$ and succeeded by another Clifford circuit $V_2$, the Clifford monomials are converted into Pauli operators as:}
        \begin{equation*}
            V_2 c_I^\dagger V_2^\dagger = \phi^*_{I_\mu} \mathbb{P}_{I_\mu}, \hspace{0.5cm}
            V_1^\dagger c_J V_1 = \phi_{J_\mu} \mathbb{P}_{J_\mu}.
        \end{equation*}
    \textit{The rest of the algorithm proceeds as before.}
\end{itemize}

Interestingly, the full circuit $V_2 U V_1$ need not be an MGC or a Clifford unitary for the modification to work. However, in a more general case of MGCs intertwined by Clifford circuits of the form $W = V_{k+1} U_k \cdots U_2 V_2 U_1 V_1$ the argument will not hold, unless it is possible to `commute out' the $V_i$'s to the left and right of the $U_i$'s.

\section{Applications}
\subsection{Efficient Matchgate Tomography}

Our discussion has focused on exploiting the compound matrix structure of $\hat{\mathcal{U}}$ in the context of verifying the action of matchgate unitaries. However, the same structure also simplifies the task of learning about them. Suppose we are given a black box unitary circuit -- which we can apply as many times as we like -- with the promise that it is composed of nearest-neighbour matchgates only. In the general case, learning about the circuit is exponentially hard in the number of qubits, requiring full gate tomography. In the matchgate case the circuit is fully characterised by its $SO(2n)$ matrix $R$, and an estimate of its $4n^2$ elements gives complete information about $U$. By Theorem \ref{thm:matchgatesupop}, these  matrix elements are precisely $\chi_\mathcal{U}(i, j)$ for $i, j \in [2n]$:

\begin{equation}
    R_{ij} = \frac{1}{2^n} \text{Tr}(c_i^\dagger U c_j U^\dagger).
\end{equation}

\noindent We can estimate the $R_{ij}$'s using steps 3--5 of Algorithm \ref{box:algorithm}. This is the Linear Optics Tomography algorithm of \citep{PRXQuantum.3.020328}. We may use it to find the quadratic Hamiltonian elements:

\begin{equation}
    \tilde{h}_{ij} = \frac{1}{4}\log(\tilde{R})_{ij},
\end{equation}

\noindent from which the estimate $\tilde{U} = e^{-i\tilde{H}}$ is computed by Equation \eqref{eq:hamiltonian}. The authors of \citep{PRXQuantum.3.020328} show that to obtain $\lVert \tilde{U} - U \rVert_\diamond \leq \epsilon$ with failure probability $\delta$, $m = \mathcal{O}(n^3 \ln(n / \delta) / \epsilon^2)$  repetitions are sufficient.

\subsection{Numerical Fidelity Estimation}

In addition to benchmarking, the efficient sampling method presented in Theorem \ref{thm:matchgate_sampling} may be applied to numerically estimate the entanglement fidelity $F_e(\mathcal{V, U})$, where $\mathcal{U}$ is an MGC and $\mathcal{V}$ is a unitary channel. To do this we may modify Algorithm \ref{box:algorithm}, with the quantities $\chi_\mathcal{V}(I, J) = 2^{-n} \mathrm{Tr}(c_I^\dagger V c_J V^\dagger)$ calculated directly from the unitary $V$. In this case, the estimator $X$ of Equation \eqref{eq:X_variable} can be calculated exactly, with expectation $\mathbb{E}(X) = F_e(\mathcal{V, U})$ and variance $\mathrm{Var}(X) \leq 1 - F_e(\mathcal{V, U})^2$ (see Appendix \ref{appendixb}). Chebyshev's inequality gives the same sample number $l = \mathcal{O}(1 / \epsilon^2 \delta)$ on the number of repetitions required to estimate $F_e(\mathcal{V, U})$ to within $\pm \epsilon$ with failure probability $\delta$. Taking into account the $\mathcal{O}(n^3)$ cost of sampling $(I, J)$ and computing $\chi_\mathcal{U}(I, J)$, the overall runtime of the estimation is given by:

\begin{equation}
t = \mathcal{O} \left( \frac{1}{\epsilon^2 \delta}( n^3 + T) \right),
\end{equation}

\noindent where $T$ is the cost of computing $\chi_\mathcal{V}(I, J)$. If $\mathcal{V}$ is a matchgate circuit intertwined by a limited number of non-matchgates, then it may be possible to bound $T$ in terms of the non-matchgate count using the techniques of \cite{mocherla2023, Reardon_Smith_2024, Dias_2024, wille2025}. We expect this may find use in the study of doped matchgate circuits \cite{paviglianiti2025, Mele_2025}. 

\section{Examples}

\subsection{Benchmarking fSim Gates}
One may use our algorithm to benchmark the fidelity of Google's~\citep{Foxen_2020} $\mathrm{fSim}(\theta, \phi)$ gates, consisting of a power-of-$i\mathrm{SWAP}$ (MG) and controlled phase (non-MG) parts:

\begin{equation}
    \mathrm{fSim}(\theta, \phi) = \begin{pmatrix}
        1 & 0 & 0 & 0 \\
        0 & \cos \theta & -i \sin \theta & 0 \\
        0 & -i \sin \theta & \cos \theta & 0 \\
        0 & 0 & 0 & e^{i \phi} 
    \end{pmatrix}.
\end{equation}

\noindent When the latter is applied with $\phi = 0$, the $\mathrm{fSim}(\theta, 0)$ operation is an $XY(\theta)$ gate. However, when $\phi \neq 0$ it becomes a $\tilde{G}(A, B)$ gate with a superoperator matrix which is no longer block diagonal. Whilst one may benchmark the $\mathrm{fSim}(\theta, 0)$ and $\mathrm{fSim}(0, \phi)$ gates separately, our method still allows a full benchmark in one sweep -- at the cost of a slightly higher expected shot count than before, dependent on the non-zero element count of the superoperator matrix (c.f. Appendix \ref{appendixb}). We calculate the sparsity to equal $94 / 256 \sim 37\%$, giving a bound on the expected value of $m$ to be:
\begin{equation}
    \mathbb{E}(m) \leq 1 + \frac{1}{\epsilon^2}( 1/\delta + 47 \ln (4/\delta) / 2).
\end{equation}
Furthermore, if $\phi = \pi$ (so that $\mathrm{fSim}(0, \pi)$ is a $\mathrm{CZ}$), by the discussion of Section \ref{sec:cliffordmg} this bound decreases. The full superoperator matrix is given in Figure \ref{fig:fsim_diagram}. 

\subsection{Numerical Simulations}

We provide data from noisy simulations of our protocol in Figure \ref{fig:simulations1}. In order to model average-to-worst case scenarios, random matchgate circuits were generated by sampling $R \in SO(2n)$ out of the Haar distribution, from which the unitary and superoperator matrices $U$, $\hat{\mathcal{U}}$ were calculated. For each choice, random Pauli SPAM circuits were generated following Algorithm \ref{box:algorithm} and simulated using Qiskit Aer to collect data. Upon each application of $U$, a noise channel followed. To model noise we used an all-qubit depolarising channel (additional simulations with all-qubit amplitude damping and all-qubit amplitude + phase damping are available in \cite{mgrepo}). Different choices of $\epsilon, \delta$ affected the average shot counts, which we found to lie well below the derived bounds. In particular, benchmarking circuits of exclusively nearest-neighbour Givens rotations led to the lowest shot counts, in agreement with Section \ref{sec:xygroup}. The same speedup also applies to nearest-neighbour $XY(\theta)$ circuits. Our plots compare the estimated entanglement fidelity $F_e(\mathcal{E, U})$ to the true value -- we find that despite the seemingly high analytic scaling, the number of repetitions required to achieve a good estimate is on par with existing approaches (which reported $\mathcal{O}(10^5)$ shots for $n=2, 3$ qubits \cite{Helsen_2022, PRXQuantum.2.010351}).

Python code and Jupyter notebooks used in our simulations are available in \cite{mgrepo}.

\acknowledgments
We thank Oscar Watts and Christopher Long for feedback on the manuscript, as well as various anonymous referees for highly insightful comments. S.S. acknowledges support from the Royal Society University Research Fellowship and ``Quantum simulation algorithms for quantum chromodynamics'' grant (ST/W006251/1) and EPSRC Reliable and Robust Quantum Computing grant (EP/W032635/1).

\appendix

\begin{figure*}[t]
    \input{fsim_render.tex}
    \caption{\textit{The Superoperator Matrix $\hat{\mathcal{U}}$ for the 2-qubit ${\mathrm{fSim}(\theta, \phi)}$ Gate.} If $\phi = 0$ it is an $XY(\theta)$ gate with a block-diagonal compound PL matrix. If $\theta = 0, \phi = \pi$ it is monomial, with entries from $\{ \pm 1, \pm i\}$. If $\theta \neq 0, \phi = \pi$, the gate is of the form $\hat{\mathcal{W}} = \hat{\mathcal{U}} \circ \hat{\mathcal{V}}$ discussed in Section \ref{sec:cliffordmg}.} 
    \label{fig:fsim_diagram}
\end{figure*}

\section{Properties of the Clifford Basis} \label{appendixa}
To prove \eqref{clifford_and_pauli}, we first note that all Clifford generators are by definition Hermitian. Consider a Clifford monomial $c_I = c_{i_1} c_{i_2} \cdots c_{i_k}$ of degree $k$, and take its Hermitian conjugate:
\begin{equation}
    c_I^\dagger = c_{i_k} \cdots c_{i_2} c_{i_1}.
\end{equation}
As the generators all anticommute, we can bring this expression back into the original order by performing a sequence of transpositions, each of which introduces a  minus sign. This requires $(k-1)$ transpositions for $c_{i_k}$, $(k-2)$ transpositions for $c_{i_{k-1}}$, and so on, the total number being the sum of an arithmetic sequence. Hence, $c_I^\dagger = (-1)^{k(k-1)/2} c_I$. It follows that each element of the Clifford basis is either Hermitian or anti-Hermitian, and this depends on whether $k(k-1)/2$ is even or odd. Starting with $k=0$ the transposition count forms the sequence of triangular numbers $(0, 0, 1, 3, 6, 10, \hdots)$, so with increasing $k$ the basis elements with $|I|=k$ switch parity every two steps. That we can always write $c_I$ as in \eqref{clifford_and_pauli} follows from the definition of the Jordan--Wigner representation; however, the above also implies that $\phi_I \times \phi_J = \pm 1$ so long as $|I| = |J|$, a fact we implicitly use in step 2 of Algorithm \ref{box:algorithm}.

\section{Runtime Parameters} \label{appendixb}
Here we derive expressions for the sample and iteration numbers $l, m_\mu$ required for our protocol. These closely follow the derivations in~\citep{Flammia_2011} with amendments. The main difference is Equation \eqref{eqn:totexpshotbound}, which explicitly links the expected shot count with superoperator sparsity.
\subsection{General Case}
\textit{Bound 1.} As discussed in the main text, we first wish to bound $Y$ to lie $\epsilon$-close to the fidelity with probability $1 - \delta$:
\begin{equation}
    \text{Pr}(|Y - F_e(\mathcal{E, U})| \geq \epsilon) \leq \delta.
\end{equation}
Much like in~\citep{Flammia_2011}, we can bound the variance of each $X_\mu$:
\begin{align}
    \text{Var}(X_\mu) &= \mathbb{E}(X_\mu^2) - \mathbb{E}(X_\mu)^2 \\
    & = \sum_\mu 2^{-2n} \chi_\mathcal{E}(I_\mu, J_\mu)^2 - F_e(\mathcal{E,U})^2 \nonumber \\
    & \leq F_e(\mathcal{E, E}) - F_e(\mathcal{E, U})^2 \leq 1. \nonumber
\end{align}
It follows that $\text{Var}(Y) \leq 1/l$, so we use Chebyshev's inequality:
\begin{equation}
    \text{Pr}(|Y - \mathbb{E}(Y)| \geq \frac{\lambda}{\sqrt{l}}) \leq \frac{1}{\lambda^2},
\end{equation}
with $\delta = 1/\lambda^2$ and $l \geq 1/(\epsilon^2 \delta)$. 

\textit{Bound 2.} We now wish to find the sample numbers $m_\mu$, such that $\tilde{Y}$ is close to $Y$:
\begin{equation}
    \text{Pr}(|\tilde{Y} - Y| \geq \epsilon) \leq \delta.
\end{equation}
To do this, we use Hoeffding's inequality:
\begin{equation}
\text{Pr}(|\tilde{Y} - \mathbb{E}(\tilde{Y})| \geq \epsilon) \leq 2e^{-2 \epsilon^2 / C}, \label{eq:hoeffding}
\end{equation}
Where $\tilde{Z}_{\mu \nu}$ are summands in $\tilde{Y} = \sum_{\mu, \nu} \tilde{Z}_{\mu \nu}$ (c.f. step 7) bounded as $a_{\mu \nu} \leq \tilde{Z}_{\mu \nu} \leq b_{\mu \nu}$, and $C = \sum_{\mu, \nu} (b_{\mu \nu} - a_{\mu \nu})^2$. Each $\tilde{Z}_{\mu \nu}$ satisfies $|\tilde{Z}_{\mu \nu}| \leq 1 / |\chi_\mathcal{U}(I_\mu, J_\mu)| m_\mu l$, therefore
\begin{align}
    C &= \sum_{\mu=1}^l \sum_{\nu=1}^{m_\mu} \frac{4}{\chi_\mathcal{U}(I_\mu, J_\mu)^2 m_\mu^2 l^2} \\
    & = \sum_{\mu=1}^l \frac{4}{\chi_\mathcal{U}(I_\mu, J_\mu)^2 m_\mu l^2}. \nonumber
\end{align}
Thus, choosing $m_\mu = \lceil 2 \ln(2/\delta) / \chi_\mathcal{U}(I_\mu, J_\mu)^2 l \epsilon^2 \rceil$ ensures that the RHS of \eqref{eq:hoeffding} equals $\delta$.

\textit{Bound 3.} Finally, we bound the expected total number of shots $\mathbb{E}(m)$, where $m = \sum_{\mu=1}^{l} m_\mu$. Since $m \ge 0$, Markov's inequality gives, for any
$\eta \in (0,1)$,
\begin{equation}
  \Pr ( m \le \tfrac{1}{\eta}\,\mathbb{E}(m) ) \ge 1 - \eta,
\end{equation}
so with probability at least $1-\eta$ the realised number of shots exceeds
$\mathbb{E}(m)$ by at most a factor $1/\eta$. We therefore bound $\mathbb{E}(m_\mu)$:
\begin{align}
    \mathbb{E}(m_\mu) &= \sum_\mu \text{Pr}(I_\mu, J_\mu) m_\mu \label{eqn:expshotbound}\\ 
    &\leq \sum_\mu \text{Pr}(I_\mu, J_\mu) \left(1 + \frac{4 \ln(4 / \delta)}{\chi_\mathcal{U}(I_\mu, J_\mu)^2 l \epsilon^2}\right) \nonumber \\
    &\leq 1 + \frac{1}{2^{2n}} \sum_\mu \frac{4 \ln(4 / \delta)}{l \epsilon^2}. \nonumber 
\end{align}
Hence $\mathbb{E}(m)$ is bounded as:
\begin{align}
    \mathbb{E}(m) &\leq l \cdot \mathbb{E}(m_\mu) \label{eqn:totexpshotbound} \\
    &\leq l + \frac{1}{2^{2n}} \sum_\mu \frac{4 \ln(4 / \delta)}{ \epsilon^2} \nonumber \\
    &\leq 1 + \frac{1}{\epsilon^2 \delta} + \frac{\#[\chi_\mathcal{U}(I,J)  \neq 0]}{2^{2n}} \frac{4 \ln(4 / \delta)}{ \epsilon^2}, \nonumber
\end{align}
Where $\#[\chi_\mathcal{U}(I,J) \neq 0]$ is the number of non-zero elements in the superoperator matrix $\hat{\mathcal{U}}$. From Theorem \ref{thm:matchgatesupop} we know this to be $\leq 2^{4n} / \sqrt{n}$ for all MGs, giving the general expectation value bound stated in the main text (we make no sharper concentration
claim; a high-probability bound matching $\mathbb{E}(m)$ up to constant factors would require controlling the variance of $\sum_\mu m_\mu$, which we leave to
future work). If $R$ has a particular structure we can bound $\mathbb{E}(m)$ further. For example, if $R$ is a 2D rotation, i.e. $R = \text{diag}(\mathbb{1}_{k-2}, R_{[2 \times 2]}, \mathbb{1}_{2n - k})$, then $\#[\chi_\mathcal{U}(I,J)  \neq 0] = \frac{3}{2}2^{2n}$, leading to a bound independent of both $n$ and $\theta$. Thus, in some cases (i.e. when $\alpha \ll 1$) it may be preferable to use the general bound to determine $l$ and $m_\mu$.

\subsection{Well-conditioned Case}

From the properties of compound matrices it follows that $\hat{\mathcal{U}}$ is an orthogonal matrix; hence, its elements $|\chi_\mathcal{U}(I, J)| \leq 1$ (in fact, this is true for all superoperators). If we are in knowledge of the well-conditioning parameter $\alpha$, such that all non-zero elements $|\chi_\mathcal{U}(I, J)| \geq \alpha$, then we can use the Hoeffding inequality for both bounds, leading to a bound for $m$ independent of system size. When finding $l$, from the above discussion it follows that the summands $Z_\mu = \frac{1}{l} X_\mu$ in $Y = \sum_\mu Z_\mu$ are bounded as $|Z_\mu| \leq \frac{1}{\alpha l}$. Hence, 
\begin{equation}
    C = \sum_{\mu=1}^l \frac{4}{\alpha^2 l^2} = \frac{4}{\alpha^2 l}. 
\end{equation}
Choosing $l = \lceil 2 \ln(2/\delta) / \alpha^2 \epsilon^2 \rceil$ then ensures the RHS of Hoeffding's inequality equals $\delta$. The expression for $m_\mu$ remains the same, and we can use the well-conditioned property again to show that $m_\mu \leq 1 + 2\ln(2/\delta)/\alpha^2 \epsilon^2 l$, and thus $m \leq m_\mu l \leq 4\ln(2/\delta)/\alpha^2 \epsilon^2$. Unlike the general case, this bound is deterministic.  

\section{Decomposition of Superoperators} \label{appendixc}
Here we give some guidance on how to calculate $\hat{\mathcal{U}}$ and identify the well-conditioning parameter $\alpha$ for modestly sized MG circuits. To do this, we make use of the following result \cite{meckes_2019}:

\begin{theorem} For any $R \in SO(n)$, there exists a set of generalised Euler angles $\{ \theta^k_j \}_{1 \leq j \leq k \leq n-1}$, where:
    \begin{align}
        1 &\leq k \leq n-1, &  1 &\leq j \leq k, 
    \end{align}
    with 
    \begin{align*}
        0 &\leq \theta^k_1 < 2 \pi, &  0 &\leq \theta^k_j < \pi \ \text{ for } j \neq 1.
    \end{align*}
    $R$ is then given by:
    \begin{equation}
        R = R^{(n-1)} \cdots R^{(1)},
    \end{equation}
    where 
    \begin{equation}
        R^{(k)} = R_1 (\theta^k_1) \cdots R_k(\theta^k_k),
    \end{equation}
    and
    \begin{equation}
        R_k(\theta) = \begin{pmatrix}
        \mathbb{1}_{k-1} &  &  &  \\
         & \cos \theta & \sin \theta &  \\
         & -\sin \theta & \cos \theta &  \\
         &  &  & \mathbb{1}_{n-k-1}
        \end{pmatrix}
    \end{equation}
    The Euler angles are unique, except when some $\theta^k_j$ is $0$ or $\pi$ (for $j \geq 2$).
\end{theorem}
\begin{proof}
We will prove this by induction. The result is true for $n=2$; assume it holds for the case of $SO(n-1)$. Then, consider $R^{(n-1)}$ acting on the $n^\text{th}$ unit vector $|n \rangle$:
\begin{align}
    R^{(n-1)} |n \rangle &= R_1(\theta^{n-1}_1) \cdots R_{n-1}(\theta^{n-1}_{n-1}) |n \rangle  \\
    &= \cos(\theta^{n-1}_{n-1}) |n \rangle \nonumber \\
    &+ \sin(\theta^{n-1}_{n-1}) \cos(\theta^{n-1}_{n-2}) |n-1 \rangle \nonumber \\
    &+ \cdots \nonumber \\
    &+ \sin(\theta^{n-1}_{n-1}) \cdots \sin(\theta^{n-1}_{2}) \cos(\theta^{n-1}_{1}) |2 \rangle \nonumber \\
    &+ \sin(\theta^{n-1}_{n-1}) \cdots \sin(\theta^{n-1}_{2}) \sin(\theta^{n-1}_{1}) |1 \rangle, \nonumber
\end{align}
Which is the unique representation of a point in $\mathbb{R}^n$ in spherical coordinates. Hence, 
\begin{equation}
    R^{(n-1)} |n \rangle = R |n \rangle,
\end{equation}
and so 
\begin{equation}
    (R^{(n-1)})^{-1} R = \begin{pmatrix} \tilde{R} & 0 \\ 0 & 1 \end{pmatrix},
\end{equation}
and by the induction hypothesis $\tilde{R}$ belongs to $SO(n-1)$. Hence, we can obtain $R$ by multiplying both sides by $R^{(n-1)}$.
\end{proof}
We can use this procedure in reverse to obtain the angles $\{ \theta^k_j\}$ for a given $R$. First, solve for $\{\theta^{n-1}_j\}$'s using the rightmost column of the matrix. Then, left-multiply by $(R^{(n-1)})^{-1} = R_{n-1}(-\theta^{n-1}_{n-1}) \cdots R_1(-\theta^1_1)$, solve for $\{\theta^{n-2}_j\}$'s, and repeat until all of the Euler angles are found. Other methods also exist, such as the algorithm of Hoffman et al. \citep{1972JMP....13..528H}.

We may also use a few important properties of compound matrices~\citep{Kravvaritis:2009aa, NAMBIAR2001251}:
\begin{theorem} \label{thm:compound_matrices} Let $A$ and $B$ be $n \times n$ matrices. The following statements are true:
    \begin{itemize}
        \item $C_r(AB) = C_r(A) C_r(B)$ (Cauchy-Binet Formula);
        \item $C_r(A^\dagger) = C_r(A)^\dagger$;
        \item If $\det A \neq 0$, then $C_r(A^{-1}) = C_r(A)^{-1}$;
        \item If $A$ is: Triangular, Diagonal, Orthogonal, Unitary, (Anti-) Symmetric, (Anti-) Hermitian, or Positive (semi-) definite, then so is $C_r(A)$.
    \end{itemize}
\end{theorem}
From the above two theorems, it follows that any MG superoperator can be written as the product of $n(2n-1)$ MG superoperators:
\begin{equation}
    \hat{\mathcal{U}} = \hat{\mathcal{U}}^{(2n-1)} \cdots \hat{\mathcal{U}}^{(1)},
\end{equation}
with each $\hat{\mathcal{U}}^{(k)}$ given by:
\begin{equation}
    \hat{\mathcal{U}}^{(k)} = \hat{\mathcal{U}}_1(\theta^k_1) \cdots \hat{\mathcal{U}}_k(\theta^k_k).
\end{equation}
Here, $\hat{\mathcal{U}}_k(\theta)$ is a block-diagonal matrix of compound matrices of $R_k(\theta)$:
\begin{equation}
    \hat{\mathcal{U}}_k(\theta) = \bigoplus_{i=0}^{2n} C_i(R_k(\theta)).
\end{equation}
Each $\hat{\mathcal{U}}_k(\theta)$ is a $2^{2n} \times 2^{2n}$ matrix, and we will show that it contains at most $\frac{3}{2}2^{2n}$ non-zero elements, which we can generate using simple rules. Furthermore, it is diagonal for $\theta = 0, \pi$, so in practice one can compute $\hat{\mathcal{U}}$ as well as $\alpha$ using the following procedure: 
\begin{enumerate}
    \item From $R$, find the Euler Angles $\{ \theta^k_j \}$.
    \item For each $\theta^k_j$, generate the non-zero elements of $\hat{\mathcal{U}}_k(\theta^k_j)$.
    \item Using sparse matrix multiplication methods, perform at most $n(2n-1)$ matrix multiplications to obtain $\hat{\mathcal{U}}$.
    \item \textit{(Optional)} Find $\alpha$ from the non-zero elements of $\hat{\mathcal{U}}$ using a sorting algorithm. 
\end{enumerate}
Although the above procedure may be used to simplify the calculation of $\hat{\mathcal{U}}$ (and optionally identify $\alpha$), it is nonetheless inefficient, relying on multiplication of an exponentially-sized matrix. Therefore, for larger systems we expect the non-well-conditioned case to be more practical, which remains scalable by Algorithm \ref{box:algorithm2}. 

We end this section with the rule for generating non-zero elements of $\hat{\mathcal{U}}_k(\theta)$:
\begin{theorem}
    The matrix elements $\chi_{\mathcal{U}_k}(I, J)$ of $\hat{\mathcal{U}}_k(\theta)$ satisfy the following:
    \begin{itemize}
        \item If $I = \emptyset$ or $I = \{ 1, \cdots, 2n \}$, then $\chi_{\mathcal{U}_k}(I, I) = 1$;
        \item If $|I|=|J|=1$, i.e. $I = \{ i \}$, $J = \{ j \}$, then $\chi_{\mathcal{U}_k}(I, J) = [R_k(\theta)]_{ij}$;
        \item If $I \ni k$ and $I \ni k+1$, or $I \niton k$ and  $I \niton k+1$, then $\chi_{\mathcal{U}_k}(I, I) = 1$;
        \item If $I \ni k$ and $I \niton k+1$, or $I \niton k$ and  $I \ni k + 1$, then $\chi_{\mathcal{U}_k}(I, I) = \cos \theta$;
        \item If $I<J$ (in lexicographic ordering), and $I, J$ differ on only one index (so that $I \ni k, I \niton k+1$ and $J \niton k, J \ni k+1$), then $\chi_{\mathcal{U}_k}(I, J) = \sin \theta$;
        \item If $I>J$, and $I, J$ differ on only one index (so that $I \niton k, I \ni k+1$ and $J \ni k, J \niton k+1$), then $\chi_{\mathcal{U}_k}(I, J) = -\sin \theta$;
        \item In all other cases, $\chi_{\mathcal{U}_k}(I, J) = 0$.
    \end{itemize}

\end{theorem}

\begin{proof}
    All cases can be easily verified by considering the different structures of submatrices formed when rows in $I, J$ are preserved, and calculating their determinants. Counting the number of non-zero elements, we obtain:
\begin{align}
    \# [\chi_{\mathcal{U}_k}(I, J) = 1] &= \frac{1}{2} 2^{2n} \\
    \# [\chi_{\mathcal{U}_k}(I, J) = \cos \theta] &= \frac{1}{2} 2^{2n} \nonumber \\
    \# [\chi_{\mathcal{U}_k}(I, J) = \pm \sin \theta] &= \frac{1}{4} 2^{2n} \nonumber
\end{align}
Giving at most $\frac{3}{2} 2^{2n}$ non-zero elements.
\end{proof}

\section{Determinantal Point Processes} \label{app:dpp}
Here, we recall the relevant properties of determinantal point processes for their use in Algorithm \ref{box:algorithm2}. We will restrict our attention to finite ground sets only. For a detailed introduction to DPPs, we recommend the review by Kulesza and Taskar \citep{Kulesza_2012}.

A \textit{determinantal point process} (DPP) on a finite ground set $[2n]$ is a probability distribution over the $2^{2n}$ random subsets $Y \subseteq [2n]$. It is characterised by a single $2n \times 2n$ \textit{marginal kernel} $K$. When $Y$ is drawn from the DPP, for every fixed $A \subseteq [2n]$, the probability that $A$ is contained in $Y$ is given by the determinant of the principal submatrix $K_{A, A}$ of $K$ indexed by $A$:

\begin{equation}
    \Pr(A \subseteq Y) = \det(K_{A, A}),
    \label{eq:dppmarginal}
\end{equation}

\noindent where $K_{A, A}$ is a submatrix of $K$ in which rows and columns corresponding to elements in $A$ are retained, just as in Theorem \ref{thm:compound_matrices}. For Equation \eqref{eq:dppmarginal} to produce valid probabilities, all principal minors of $K$ must be non-negative. Typically, this is ensured by requiring $K$ to be a real, symmetric and positive semidefinite matrix with eigenvalues in $[0, 1]$. DPPs have a natural connection to fermionic systems \cite{Macchi1975TheCA, Cunden_2018}, as the determinant structure encodes negative correlations reminiscent of the Pauli exclusion principle. For example, for a pair of elements $i, j \in [2n]$, the probability that both are contained in $Y$ is given by $\Pr(i, j \in Y) = K_{ii}K_{jj} - K_{ij}^2$, which is always less than or equal to the product of individual probabilities $\Pr(i \in Y)\Pr(j \in Y) = K_{ii}K_{jj}$. Thus, correlations between elements are always non-positive and `spread-out' subsets are favoured.

Typically, the probability of exactly drawing a subset $Y$ is not given directly by Equation \eqref{eq:dppmarginal}, but instead requires an $L$-\textit{ensemble} for which the atomic probabilities $\Pr(Y = J) \propto \det(L_{J,J})$, with $L$ a real, symmetric matrix related to $K$. In our case, we will not need to construct an $L$-ensemble, as our kernel $K^{(I)}= R_{I, [2n]}^T R_{I, [2n]}$ defines an \textit{elementary} DPP. This is one of the simplest instances of a determinantal point process, obeying the following properties \cite{Kulesza_2012}:
\begin{itemize}
    \item The kernel $K$ satisfies $K = \sum_{\mathbf{v} \in \mathcal{V}} \mathbf{v} \mathbf{v}^T$ for some set of orthonormal vectors $\mathcal{V}$, and is therefore a projection matrix with eigenvalues in $\{0, 1\}$.
    \item If $Y$ is drawn from the elementary DPP, then:
    \begin{equation}
    \Pr(|Y| = k) = 1, \nonumber
    \end{equation} 
    \noindent where $k = |\mathcal{V}| = \mathrm{rank}(K) = \mathrm{Tr}(K)$.
\end{itemize}

\noindent Because for any $Y$ drawn from the elementary DPP we have that $|Y| = k$ almost surely, Equation \eqref{eq:dppmarginal} collapses to the atomic probability $\Pr(Y = J) = \det(K^{(I)}_{J, J})$ referred to in Equation \eqref{eq:dpp_determinant}. Therefore, using a DPP sampling algorithm with the kernel $K^{(I)}$ will sample $J \sim P(J | I, k)$ with the correct target weights $|\chi_{\mathcal{U}}(I, J)|^2$. Constructing the kernel $K^{(I)}$ for an index $I$ sampled in the first stage of Algorithm \ref{box:algorithm2} is efficient, requiring $\mathcal{O}(n^3)$ time. Similarly, the DPP sampling algorithm of Hough, Krishnapur, Peres and Vir\'ag \cite{Hough_2006} runs in $\mathcal{O}(n^3)$ time. We now present the DPP sampling algorithm.

First, set $J = \emptyset$ and pick an orthonormal basis $\mathrm{ONB}(\mathcal{V}) = \{|v_1\rangle, \dots, |v_k\rangle\}$ from the range of $K^{(I)}$. This can be done by taking the $k$ rows of $R_{I, [2n]}$ which are orthonormal by construction. Begin by drawing an element $y \in [2n]$ with probability $\Pr(y) = \frac{1}{\dim \mathcal{V}} \sum_{|v_i \rangle \in \mathrm{ONB}(\mathcal{V})} | \langle v_i | y \rangle |^2$, where $|y \rangle$ is a standard unit vector. Then, set $J \leftarrow J \cup \{y\}$ and replace $\mathcal{V}$ by the orthogonal complement of $|y \rangle$ in $\mathcal{V}$ using the Gram--Schmidt procedure, reducing $\dim \mathcal{V}$ by one. Repeat this process until $\dim \mathcal{V} = 0$, at which point $J$ will contain $k$ elements drawn from the DPP with target weights $\det(K^{(I)}_{J,J})$. The overall complexity of this algorithm is $\mathcal{O}(n^3)$, dominated by the Gram--Schmidt orthogonalisation steps. This is summarised in Algorithm \ref{box:algorithm3}; for a proof of its correctness see \cite{Hough_2006,Kulesza_2012}.

\

\begin{tcolorbox}[title = Algorithm 3: Elementary DPP Sampling, colback=white, colframe=camblue, fonttitle=\bfseries, title filled=false]
    \begin{algorithm}[H]
        \DontPrintSemicolon
        1. Set $\mathcal{V} \leftarrow \mathrm{Im}(K^{(I)})$, $J \leftarrow \emptyset$. \\
        \While{$\dim \mathcal{V} = t > 0$} {
            2. Draw $y \in [2n]$ with probability:
            \begin{equation}
                p(y) = \frac{1}{t} \sum_{|v_i \rangle \in \mathrm{ONB}(\mathcal{V})} | \langle v_i | y \rangle |^2. \nonumber \end{equation} \\
            3. Set $J \leftarrow J \cup \{y\}$. \\
            4. Set $\mathcal{V} \leftarrow \mathcal{V}_\perp$, where $\mathcal{V}_\perp$ is the subspace of $\mathcal{V}$ orthogonal to $|y \rangle$.}
        \Return $J$.
    \label{box:algorithm3}{}
    \end{algorithm}
\end{tcolorbox}

\bibliography{bibliography.bib}

\end{document}

%% file: fsim_render.tex
\resizebox{1\textwidth}{!}{$
\begin{array}{c|c|cccc|cccccc|cccc|c}
    & (0) & (1) & (2) & (3) & (4) & (12) & (13) & (14) & (23) & (24) & (34) & (123) & (124) & (134) & (234) & (1234) \\
    \hline
    (0) & 1 & 0 & 0 & 0 & 0 & 0 & 0 & 0 & 0 & 0 & 0 & 0 & 0 & 0 & 0 & 0 \\
    \hline 
    (1) & 0 & \cos \theta  \cos^2 \frac{\phi}{2}  & - \frac{1}{2} \cos \theta \sin \phi & \frac{1}{2} \sin \theta \sin \phi & \sin \theta \cos^2 \frac{\phi}{2}& 0 & 0 & 0 & 0 & 0 & 0 & -\frac{1}{2} i \sin \theta \sin \phi & i \sin \theta \sin ^2 \frac{\phi }{2} & i \cos \theta \sin ^2 \frac{\phi }{2} & \frac{1}{2} i \cos \theta \sin \phi & 0 \\
    (2) & 0 & \frac{1}{2} \cos \theta \sin \phi & \cos \theta  \cos^2 \frac{\phi}{2} & -\sin \theta \cos^2 \frac{\phi}{2} & \frac{1}{2} \sin \theta \sin \phi & 0 & 0 & 0 & 0 & 0 & 0 & -i \sin \theta \sin^2 \frac{\phi}{2} & -\frac{1}{2} i \sin \theta \sin \phi & -\frac{1}{2} i \cos \theta \sin \phi & i \cos \theta \sin ^2 \frac{\phi }{2} & 0 \\
    (3) & 0 & \frac{1}{2} \sin \theta \sin \phi & \sin \theta \cos^2 \frac{\phi}{2} & \cos \theta  \cos^2 \frac{\phi}{2} & -\frac{1}{2} \cos \theta \sin \phi & 0 & 0 & 0 & 0 & 0 & 0 & i \cos \theta \sin ^2 \frac{\phi }{2} & \frac{1}{2} i \cos \theta \sin \phi & -\frac{1}{2} i \sin \theta \sin \phi & i \sin \theta \sin ^2 \frac{\phi }{2} & 0 \\
    (4) & 0 & -\sin \theta \cos^2 \frac{\phi}{2} & \frac{1}{2} \sin \theta \sin \phi & \frac{1}{2} \cos \theta \sin \phi & \cos \theta  \cos^2 \frac{\phi}{2} & 0 & 0 & 0 & 0 & 0 & 0 & -\frac{1}{2} i \cos \theta \sin \phi & i \cos \theta \sin ^2 \frac{\phi }{2} & -i \sin \theta \sin^2 \frac{\phi}{2} & -\frac{1}{2} i \sin \theta \sin \phi & 0 \\
    \hline
    (12) & 0 & 0 & 0 & 0 & 0 & \cos ^2\theta & -\sin \theta \cos \theta & 0 & 0 & -\sin \theta \cos \theta & \sin ^2\theta & 0 & 0 & 0 & 0 & 0 \\
    (13) & 0 & 0 & 0 & 0 & 0 & \sin \theta \cos \theta & \frac{1}{2} (\cos 2 \theta +\cos \phi) & -\frac{1}{2}\sin\phi & -\frac{1}{2}\sin\phi & \frac{1}{2} (\cos 2 \theta -\cos \phi) & -\sin \theta \cos \theta & 0 & 0 & 0 & 0 & 0 \\
    (14) & 0 & 0 & 0 & 0 & 0 & 0 & \frac{1}{2}\sin\phi & \cos ^2 \frac{\phi }{2} & -\sin^2 \frac{\phi}{2} & -\frac{1}{2}\sin\phi & 0 & 0 & 0 & 0 & 0 & 0 \\
    (23) & 0 & 0 & 0 & 0 & 0 & 0 & \frac{1}{2}\sin\phi & -\sin^2 \frac{\phi}{2} & \cos ^2 \frac{\phi }{2} & -\frac{1}{2}\sin\phi & 0 & 0 & 0 & 0 & 0 & 0 \\
    (24) & 0 & 0 & 0 & 0 & 0 & \sin \theta \cos \theta & \frac{1}{2} (\cos 2 \theta -\cos \phi) & \frac{1}{2}\sin\phi & \frac{1}{2}\sin\phi & \frac{1}{2} (\cos 2 \theta +\cos \phi) & -\sin \theta \cos \theta & 0 & 0 & 0 & 0 & 0 \\
    (34) & 0 & 0 & 0 & 0 & 0 & \sin ^2\theta & \sin \theta \cos \theta & 0 & 0 & \sin \theta \cos \theta & \cos ^2\theta & 0 & 0 & 0 & 0 & 0 \\
    \hline
    (123) & 0 & \frac{1}{2} i \sin \theta \sin \phi & -i \sin \theta \sin^2 \frac{\phi}{2} & -i \cos \theta \sin^2 \frac{\phi}{2} & -\frac{1}{2} i \cos \theta \sin \phi & 0 & 0 & 0 & 0 & 0 & 0 & \cos \theta  \cos^2 \frac{\phi}{2} & -\frac{1}{2} \cos \theta \sin \phi & \frac{1}{2} \sin \theta \sin \phi & \sin \theta \cos^2 \frac{\phi}{2} & 0 \\
    (124) & 0 & i \sin \theta \sin ^2 \frac{\phi }{2} & \frac{1}{2} i \sin \theta \sin \phi & \frac{1}{2} i \cos \theta \sin \phi & -i \cos \theta \sin^2 \frac{\phi}{2} & 0 & 0 & 0 & 0 & 0 & 0 & \frac{1}{2} \cos \theta \sin \phi & \cos \theta  \cos^2 \frac{\phi}{2} & -\sin \theta \cos^2 \frac{\phi}{2} & \frac{1}{2} \sin \theta \sin \phi & 0 \\
    (134) & 0 & -i \cos \theta \sin^2 \frac{\phi}{2} & -\frac{1}{2} i \cos \theta \sin \phi & \frac{1}{2} i \sin \theta \sin \phi & -i \sin \theta \sin^2 \frac{\phi}{2} & 0 & 0 & 0 & 0 & 0 & 0 & \frac{1}{2} \sin \theta \sin \phi & \sin \theta \cos^2 \frac{\phi}{2} & \cos \theta  \cos^2 \frac{\phi}{2} & -\frac{1}{2} \cos \theta \sin \phi & 0 \\
    (234) & 0 & \frac{1}{2} i \cos \theta \sin \phi & -i \cos \theta \sin^2 \frac{\phi}{2} & i \sin \theta \sin ^2 \frac{\phi }{2} & \frac{1}{2} i \sin \theta \sin \phi & 0 & 0 & 0 & 0 & 0 & 0 & -\sin \theta \cos^2 \frac{\phi}{2} & \frac{1}{2} \sin \theta \sin \phi & \frac{1}{2} \cos \theta \sin \phi & \cos \theta  \cos^2 \frac{\phi}{2} & 0 \\
    \hline
    (1234) & 0 & 0 & 0 & 0 & 0 & 0 & 0 & 0 & 0 & 0 & 0 & 0 & 0 & 0 & 0 & 1 \\
    \end{array}
    $}